\documentclass[12pt]{article}

\usepackage{geometry}
\usepackage{yfonts}
\usepackage{subcaption}

\usepackage{amsthm}
\RequirePackage{etex}

\newcounter{mycou}

\usepackage{pict2e}

\makeatletter
\newcommand{\CCOp}{\mathord{\mathpalette\nicoud@YESNO{\nicoud@path{\fillpath}}}}
\newcommand{\nicoud@YESNO}[2]{%
  \begingroup
  \settoheight{\unitlength}{$#1X$}%
  \begin{picture}(0.7,1)
  \linethickness{\variable@rule{#1}}%
  \roundcap\roundjoin
  \nicoud@path{\strokepath}
  #2
  \Line(0.35,-0.35)(0.35,1.08)
  \Line(0.55,-0.35)(0.55,1.08)
  \end{picture}%
  \endgroup
}
\newcommand{\nicoud@path}[1]{%
  \moveto(0.5,0.9)
  \lineto(0.5,1)\lineto(0.6,1)\lineto(0.6,0.9)
  \closepath
  \moveto(0.3,0.9)
  \lineto(0.3,1)\lineto(0.4,1)\lineto(0.4,0.9)
  \closepath
  \moveto(0.5,-0.27)
  \lineto(0.5,-0.17)\lineto(0.6,-0.17)\lineto(0.6,-0.27)
  \closepath
  \moveto(0.3,-0.27)
  \lineto(0.3,-0.17)\lineto(0.4,-0.17)\lineto(0.4,-0.27)
  \closepath
  #1
}
\newcommand{\variable@rule}[1]{%
  \fontdimen8  
  \ifx#1\displaystyle\textfont3\else
    \ifx#1\textstyle\textfont3\else
      \ifx#1\scriptstyle\scriptfont3\else
        \scriptscriptfont3\relax
  \fi\fi\fi
}
\makeatletter

\makeatletter
\newcommand{\subalign}[1]{%
  \vcenter{%
    \Let@ \restore@math@cr \default@tag
    \baselineskip\fontdimen10 \scriptfont\tw@
    \advance\baselineskip\fontdimen12 \scriptfont\tw@
    \lineskip\thr@@\fontdimen8 \scriptfont\thr@@
    \lineskiplimit\lineskip
    \ialign{\hfil$\m@th\scriptstyle##$&$\m@th\scriptstyle{}##$\hfil\crcr
      #1\crcr
    }%
  }%
}
\makeatother

 \newtheoremstyle{theoremdd}
  {0}
  {0}
  {\itshape}
  {0pt}
  {\bfseries}
  {}
  { }
  {\thmname{#1}\thmnumber{ #2}\textnormal{\thmnote{ (#3)}}}

\theoremstyle{theoremdd}

\usepackage{setspace,multirow}

\makeatletter
\newcommand*{\rom}[1]{\expandafter\@slowromancap\romannumeral #1@}
\makeatother

\usepackage{lmodern,enumerate,rotating,anyfontsize}
\usepackage{amsmath,amsfonts,amssymb,mathtools,blkarray,soul,etoolbox}
\usepackage{tikz-cd}
\usepackage{tikz-network}

\makeatletter
\newcommand{\mylabel}[2]{#2\def\@currentlabel{#2}\label{#1}}
\makeatother

\usetikzlibrary{cd}
\usetikzlibrary{shapes,shapes.geometric,arrows,fit,calc,positioning,automata}
\usetikzlibrary {circuits.logic.IEC}

\usepackage{stmaryrd}
\SetSymbolFont{stmry}{bold}{U}{stmry}{m}{n}

\definecolor{green}{rgb}{0.1,0.7,0.1}

\DeclareMathOperator{\CC}{CC}

\DeclareMathOperator{\Obs}{Obs}

\DeclareMathOperator{\obs}{obs}

\DeclareMathOperator{\inv}{inv}

\allowdisplaybreaks

\newcommand{\N}{\mathbb{N}}

\newcommand{\dt}{\delta}

\newcommand{\ep}{\epsilon}

\newcommand{\Scal}{\mathcal{S}}

\newcommand{\Sig}{\Sigma}

\newcommand{\Ycal}{\mathcal{Y}}

\newcommand{\Mt}{\mathcal{M}}

\newcommand{\llb}{\llbracket}
\newcommand{\rrb}{\rrbracket}

\newcommand{\Ysf}{\mathsf{Y}}
\newcommand{\ysf}{\mathsf{y}}

\newtheorem{theorem}{Theorem}[section]
\newtheorem{definition}{Definition}

\newtheorem{remark}{Remark}

\newtheorem{fact}{Fact}

\usepackage{times,amssymb,amsfonts,mathtools,graphicx,tikz,algorithm,algorithmic,url,hhline,booktabs}
\usetikzlibrary{automata,arrows,positioning}

\usepackage[english]{babel}
\addto\captionsenglish{%
}
\addto\extrasenglish{%
}

\usepackage{aliascnt}

\theoremstyle{plain}

\newaliascnt{lemma}{theorem}

\aliascntresetthe{lemma}

\newaliascnt{proposition}{theorem}

\aliascntresetthe{proposition}

\newaliascnt{corollary}{theorem}

\aliascntresetthe{corollary}

\newaliascnt{example}{theorem}
\newtheorem{example}[example]{Example}
\aliascntresetthe{example}

\usetikzlibrary{shadows}

\makeatletter
\let\NAT@parse\undefined
\makeatother
\usepackage[colorlinks,linkcolor=blue,anchorcolor=blue,citecolor=green,urlcolor=cyan,hyperindex]{hyperref}

\newcommand{\PSPACE}{\mathsf{PSPACE}}

\newcommand{\EXPTIME}{\mathsf{EXPTIME}}

\usetikzlibrary{automata,arrows,positioning}
\usetikzlibrary{petri}

\tikzset{elliptic state/.style={draw, ellipse, thick, fill=gray!10}
}

\tikzset{rectangular state/.style={draw, rectangle, thick, fill=gray!10}
}

\tikzset{emptystate/.style={}
}

\tikzset{
    partial ellipse/.style args={#1:#2:#3}{
        insert path={+ (#1:#3) arc (#1:#2:#3)}
    }
}

\tikzset{
node distance=3cm, 
every state/.style={thick, fill=gray!10}, 
initial text=$ $, 
}



\usepackage{tcolorbox}
\usepackage{tabularx}
\usepackage{array}
\usepackage{colortbl}
\tcbuselibrary{skins}

\title{Order-2 bygone-state opacity of\\ labeled finite-state automata}

\author{Kuize Zhang\\
{\small School of Mathematics and Statistics}\\
{\small Xi'an Jiaotong University, 810049 Xi'an, China}\\
{\small kuize.zhang@xjtu.edu.cn}
}

\begin{document}

\date{}

\maketitle

{\bf Abstract}
  In this paper, we formulate a scenario that an agent can never be sure that another agent can uniquely determine the state of a finite-state automaton based on its observations to the automaton at the current and any past time as the property of order-$2$ bygone-state opacity. Based on our concurrent composition and the classical observer, we derive a tool to verify this property in doubly exponential time. The interest of this result lies in that we extend inference of finite automata from a single agent to two ordered agents.


\tableofcontents

\section{Introduction}
\label{sec:intro}

In this paper, we consider the scenario consisting of a \emph{finite-state automaton} (FSA) $G$ and
$2$ \emph{agents} $A_1$ and $A_2$.
The FSA is publicly known to both agents, each agent $A_i$ can observe a subset $E_i$ of events
of $G$ via a labeling function $\ell_i: E_i\to \Sig_i$ with $\Sig_i$ an alphabet,
and agent $A_2$ knows agent $A_1$'s observable events and labeling function. 
We formulate and compute what $A_2$ knows about $A_1$'s \emph{bygone-state estimate} of $G$ based on an observed label sequence generated by $G$. 
We also formulate such a property as follows: for every trace generated by $G$,
if $A_1$ can uniquely determine $G$'s consistent state at a past time according to the label sequence generated by the trace,
then $G$ can generate another trace such that based on the label sequence generated by 
the second trace $A_1$ cannot uniquely determine $G$'s
consistent state at a past time, and the two traces look the same to agent $A_2$ 
before and after the respective time instants $A_1$ is doing state estimations
(see \autoref{fig1:Ord2BGStateOpacity} as an illustration).
Roughly
speaking, $A_2$ cannot be sure whether $A_1$ can uniquely determine $G$'s state according
to its inferences of $A_1$'s bygone-state estimates of $G$ based on $A_2$'s observation
to $G$.
Such a property will be called \emph{order-$2$ bygone-state opacity}. When there is 
only one agent, this property degenerates to the infinite-step opacity
\cite{Saboori2012InfiniteStepOpacity}: $A_1$ cannot uniquely determine $G$'s state 
according to every observation to $G$ at every past time. The infinite-step opacity
verification problem is $\PSPACE$-complete in labeled finite-state automata (LFSAs)
\cite{Saboori2012InfiniteStepOpacity}
and currently the most efficient algorithm for verifying this property is based on the combination
of the concurrent composition \cite{Zhang2020DetPNFA} and the observer
(that is, the classical powerset construction \cite{RabinScott1959PowersetConstruction}),
see \cite{Zhang2023UnifiedFrame4DES} for details.

  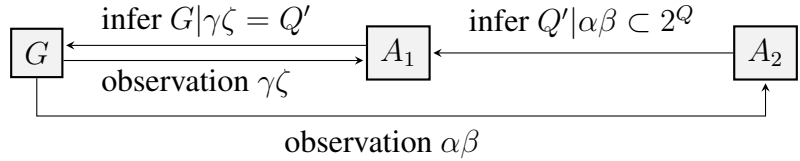
\begin{figure}[!htbp]
	\centering
	\begin{tikzpicture}[>=stealth',shorten >=1pt,auto,node distance=4.0 cm, scale = 1.0, transform shape,  >=stealth,inner sep=5pt, every text node part/.style={align=center}]

	  \node[rectangular state] (plant1) {$G$};
	  \node[rectangular state, right = 4cm of plant1] (Usr1) {$A_1$};
	  \node[rectangular state, right = 4cm of Usr1] (Intr1) {$A_2$};

	  \path [->]
	  (Intr1) edge node [above, sloped] {infer $Q'|\alpha\beta\subset 2^Q$} (Usr1)
	  ;

	  \draw [->] ($(plant1.east)+(0,-0.1cm)$) -- ++(0,0) node [below, sloped, xshift=50, yshift=2] {observation $\gamma\zeta$} -- ($(Usr1.west)+(0,-0.1cm)$);
	  \draw [->] ($(Usr1.west)+(0,0.1cm)$) -- ++(0,0) node [above, sloped, xshift=-60, yshift=-2] {infer $G|\gamma\zeta=Q'$} -- ($(plant1.east)+(0,0.1cm)$);
	  \draw [->] (plant1.south) -- ++(0,-0.5) node [below, xshift=130] {observation $\alpha\beta$} -| (Intr1.south)
	  ;
    \end{tikzpicture}
	\caption{An illustration of order-$2$ bygone-state opacity: 
	  agent $A_2$ observes label sequence $\alpha\beta$ from $G$, infers all possible
	  traces $s_1s_2$ such that $s_1$ and $s_2$ generate label sequences $\alpha$ and $\beta$,
	  then agent $A_1$ observes $\gamma$ and $\zeta$ from $s_1$ and $s_2$ and computes state estimate $Q'$ of $G$ based on $\gamma$ and $\zeta$. All such $Q'$ form 
	$A_2$'s inferences of $A_1$'s state estimates based on $\alpha\beta$.}
	\label{fig1:Ord2BGStateOpacity} 
  \end{figure}

The main results of the paper are twofold: we formulate order-$2$ bygone-state opacity,
and derive a tool to verify this property in $2$-$\EXPTIME$ based on our concurrent
composition \cite{Zhang2020DetPNFA} and the classical observer \cite{RabinScott1959PowersetConstruction}.

\section{Preliminaries}
\label{sec:prelim}

\subsection{Notation}
\label{subsec:notation}

{\bf Notation.} 
Symbol $\N$ denotes the set of nonnegative integers.
For two nonnegative integers $i\le j$, $\llb i,j\rrb$ denotes the set of all integers no less than $i$ and no greater 
than $j$.
For a set $S$, $|S|$ denotes its cardinality and $2^S$ its power set. Symbols $\subset$,
$\not\subset$, and $\subsetneq$ denote the \emph{subset of}, \emph{not a subset of}, and
\emph{a strict subset of} relations, respectively. For two sets $A$ and $B$, $A\setminus B$ denotes
$\{x\in A|x\notin B\}$.
Let $\Sigma$ denote a finite \emph{alphabet}. 
Elements of $\Sigma$ are called \emph{letters}. As usual, 
$\Sig^*$ denotes the set of \emph{words} or \emph{strings} (i.e., finite-length sequences of letters)
over $\Sig$ including the empty word $\epsilon$, 
and denote $\Sig^{+}:=\Sig^*\setminus\{\epsilon\}$.
The \emph{length} of a word $w\in\Sig^*$ is the number of letters, counting repetitions, occurring in $w$,
and is denoted by $|w|$.
For a word $w\in \Sig^*$, its \emph{mirror image} $w^R$ is obtained by reversing the order
of the letters of $w$. For a string $(a_1,b_1) \dots (a_n, b_n) =: s$, $a_1 \dots a_n =: s(L)$, $b_1\dots b_n =:
s(R)$.
In this paper, we consider $2$ agents, and use alphabets $\Sig_i$ to denote the label sets of
agent $A_i$, $i=1,2$.
A \emph{(formal) language} is a subset of $\Sig^*$.
For two languages $L_1,L_2\subset\Sig^*$, their \emph{concatenation} $L_1L_2$ is defined as $\{e_1e_2|e_1\in L_1,
e_2\in L_2\}$. 
For two formal languages $L_1,L_2\subset\Sig^*$, denote their \emph{interleaving}
as $L_1 \interleave L_2 = \{e_1\dots e_n| (\exists\text{ subsequences }i_1\dots i_r,
j_1,\dots,j_s)[(i_1<\cdots< i_r) \wedge (j_1<\cdots<j_s) \wedge (r+s=n) \wedge
\{i_1,\dots,i_r,j_1,\dots,j_s\} = \{1,\dots,n\} \wedge e_{i_1}\dots e_{i_r}\in L_1
\wedge e_{j_1}\dots e_{j_s}\in L_2] \}$. Hence each word of $L_1\interleave L_2$ is obtained
by interleaving a word of $L_1$ and a word of $L_2$. 

\vspace{-0.2cm}
\noindent\rule{\linewidth}{0.4pt}
\vspace{-0.1cm}
\noindent\textbf{Symbols} used throughout the paper:

\vspace{-0.2cm}
\noindent\rule{\linewidth}{0.4pt}   

\indent $G=(Q,E,\dt,Q_0)$: finite-state automaton

$\Scal=(G,\Sig,\ell)$: labeled finite-state automaton

$G^{\inv}$: inverse finite-state automaton

$L(G)$: language generated by $G$ 

$\Mt(\Scal,\alpha) = \Mt_{\ell}(G,\alpha)$: current-state estimate of $G$ with respect to label sequence $\alpha$

$\Mt(\Scal,\alpha,\alpha\beta) = \Mt_{\ell}(G,\alpha,\alpha\beta)$: bygone-state estimate of $G$ with respect to label sequences $\alpha$ and $\alpha\beta$

$\CC(\Scal^1,\Scal^2)$: concurrent composition of labeled finite-state automata $\Scal^1$ and $\Scal^2$

$\CC_{\ell}(G^1,G^2)$: concurrent composition of finite-state automata $G^1$ and $G^2$ endowed with the same labeling function $\ell$

$\Obs(\Scal)=\Obs_{\ell}(G)$: observer of labeled finite-state automaton $\Scal=(G,\Sig,\ell)$

$\Obs^{\inv}(\Scal)=\Obs_{\ell}^{\inv}(G)$: inverse observer of labeled finite-state automaton $\Scal=(G,\Sig,\ell)$

$Y_1\Cap Y_2$ with $Y_1,Y_2\subset Q\times 2^Q$: $\left\{ X_1\cap X_2 | (\exists q\in Q) [(q,X_1)\in Y_1 \wedge (q,X_2)\in Y_2]\right\}$ (particularly if 
	such a state $q$ does not exist, then $Y_1\Cap Y_2=\{\emptyset\}$.)

\vspace{-0.2cm}
\noindent\rule{\linewidth}{0.4pt}   
\vspace{-0.2cm}
\indent $(G,\Sig_i,\ell_i)=:G_{A_i}$: finite-state automaton $G$ and agent $A_i$ endowed with labeling function $\ell_i:E\to\Sig_i$

$\Mt_{\ell_i}(G,\alpha)=:\Mt_i(G,\alpha)$: current-state estimate of $G$ done by agent $A_i$ with respect to label sequence $\alpha$

$\Mt_{\ell_i}(G,\alpha,\alpha\beta) =: \Mt_i(G,\alpha,\alpha\beta)$: bygone-state estimate of $G$ done by agent $A_i$ with respect to label sequences $\alpha$ and $\alpha\beta$

$\Obs_{\ell_i}(G) =: \Obs_{A_i}$: observer of $G$ with respect to agent $A_i$

$\Obs_{\ell_i}^{\inv}(G) =: \Obs_{A_i}^{\inv}$: inverse observer of $G$ with respect to agent $A_i$

$\Mt_{A_1\leftarrow A_2}(G,\alpha)$: order-$2$ current-state estimate of $G$ with respect to agents $A_1$ and $A_2$ and $A_2$'s label sequence $\alpha$

$\Mt_{A_1\leftarrow A_2}(G,\alpha,\alpha\beta)$: order-$2$ bygone-state estimate of $G$ with respect to agents $A_1$ and $A_2$ and $A_2$'s label sequences $\alpha$ and $\alpha\beta$

$\Obs_{A_1\leftarrow A_2}(G)$: order-$2$ observer

$\Obs_{A_1\leftarrow A_2}^{\inv}(G)$: order-$2$ inverse observer

$\CC_{\ell_{\ep}}\left(\Obs_{A_1\leftarrow A_2} (G), \Obs_{A_1\leftarrow A_2}^{\inv} (G) \right)$: concurrent composition of order-$2$ observer $\Obs_{A_1\leftarrow A_2} (G)$ and order-$2$ inverse observer $\Obs_{A_1\leftarrow A_2}^{\inv} (G)$ endowed with the $\ep$-labeling function

\vspace{-0.2cm}
\noindent\rule{\linewidth}{0.4pt}   

\vspace{-0.2cm}

\subsection{Labeled finite-state automata}

\begin{definition}
  A \emph{finite-state automaton} (FSA) is a quadruple 
  \begin{align}\label{FSA}
	G=(Q,E,\dt,Q_0),
  \end{align} where
  \begin{enumerate}
  	\item $Q$ is a finite set of \emph{states},
    \item $E$ is an alphabet of \emph{events},
	\item $\dt:Q\times E \to 2^Q$ is the \emph{transition function} (equivalently described as
		$\dt\subset Q\times E\times Q$ such that $(q,e,q')\in\dt$ if and only if $q'\in\dt(q,e)$),
	\item $Q_0\subset Q$ is a set of \emph{initial states}.
  \end{enumerate}
\end{definition}

\begin{definition}
  The \emph{inverse finite-state automaton}, denoted as $G^{\inv}$, of FSA $G$ as in \eqref{FSA} is a quadruple 
  \begin{align}\label{inv-FSA}
	(Q,E,\dt^{\inv},Q),
  \end{align} where
  \begin{enumerate}
	\item for all $q,q'\in Q$ and $e\in E$, $q'\in\dt^{\inv}(q,e)$ if and only if $q\in\dt(q',e)$,
	\item $Q$ is the initial state set.
  \end{enumerate}
\end{definition}

Transition function $\dt$ is recursively extended to $Q\times E^*\to 2^Q$: for all $q\in Q$, $u\in E^*$, and $e\in E$,
$\dt(q,\ep)=\{q\}$, $\dt(q,ue)=\bigcup_{p\in\dt(q,u)}\dt(p,e)$. 
Automaton $G$ is called \emph{deterministic} if $Q_0=\{q_0\}$ for some $q_0\in Q$, and
for all $q\in Q$ and $e\in E$, $|\dt(q,e)|\le 1$. A deterministic FSA $G$ is also 
denoted as $G=(Q,E,\dt,q_0)$, and in this case, $\dt$ is a partial function from $Q\times E$ to $Q$.

A transition $q\xrightarrow[]{e} q'$ with $q'\in \dt(q,e)$ means that when $G$ is in state $q$ and event 
$e$ occurs, $G$ transitions to state $q'$. A sequence $q_0\xrightarrow[]{e_1}\cdots \xrightarrow[]{e_n}
q_n$ of consecutive transitions with $n\in\N$ is called a \emph{run}\footnote{When $n=0$, the run degenerates
to a single state $q_0$.},
in which the event sequence $e_1\dots e_n$ is called
a \emph{trace (generated by $G$)} if $q_0\in Q_0$. A state $q\in Q$ is \emph{reachable} if there is a run
from some initial state to $q$. The \emph{reachable part} of $G$ consists of all
reachable states and transitions between them.
When showing an automaton, we usually only show its reachable part (e.g., \autoref{FA:fig27})
or its necessary part (e.g., \autoref{fig2:High-OrderCSO} and \autoref{fig4:High-OrderCSO})
when the automaton is large.
The \emph{language} $L(G)$ \emph{generated by $G$} is the set of traces generated by $G$. 

Occurrences of events of an FSA $G$ may be observable or not. Let alphabet $\Sig$ denote the set
of \emph{labels/outputs}. The \emph{labeling function} is defined as $\ell:E\to \Sig\cup\{\ep\}$, and 
is recursively extended to $\ell:E^*\to \Sig^*$. 
Particularly, denote $\ell_{\ep}:E\to\{\ep\}$ and call it \emph{$\ep$-labeling function}.
Denote $\ell(L(G))$ as the set of label sequences generated by $G$ with respect to labeling function $\ell$.
Denote $E_o=\{e\in E|\ell(e)\in \Sigma\}$,
$E_{uo}=\{e\in E|\ell(e)=\ep \}$, where the former denotes the set of \emph{observable events} and the latter 
denotes the set of \emph{unobservable events}. When an observable $e$ occurs; its label $\ell(e)$ is observed,
when an unobservable event occurs, nothing is observed. An \emph{observable transition} (resp., 
\emph{unobservable transition}) is a transition whose event is observable (resp., unobservable).
Two runs are called \emph{observationally equivalent with respect to a labeling function $\ell$} if their event sequences generate the same label sequence under $\ell$.
Sometimes we simply say two runs are observationally equivalent if $\ell$ is clear from the 
context.
Anologously, \emph{observationally equivalent event sequences} can be defined.


A \emph{labeled finite-state automaton} (LFSA) is denoted as 
\begin{equation}\label{LFSA}
  \Scal=(G,\Sig,\ell).
\end{equation}
Its inverse LFSA is $\Scal^{\inv}=(G^{\inv},\Sig,\ell)$.

The following definition of state estimate is critical to define all kinds of properties in LFSAs.
The \emph{current-state estimate} $\Mt(\Scal,\alpha)$ with respect to $\alpha\in\Sig^*$ is defined as
\begin{subequations}\label{eqn:CSE}
\begin{align}
  &\Mt(\Scal,\alpha)=\Mt_{\ell}(G,\alpha) \\
  =& \{q\in Q|(\exists\text{ run }q_0\xrightarrow[]{s} q)[q_0\in Q_0\wedge \ell(s)=\alpha]\},
\end{align}
\end{subequations}
which is the set of states of $G$ may be in when $\alpha$ is observed.
Evidently, $\Mt(\Scal,\alpha)=\emptyset$ when $\alpha\in\Sig^*\setminus \ell(L(G))$.

We formulate \emph{bygone-state estimate $\Mt(\Scal,\alpha,\alpha\beta)$ of $G$ with respect to $\alpha\beta\in\ell(L(G))$} as follows:
\begin{subequations}\label{eqn:ByGSE} 
\begin{align}
  &\Mt(\Scal,\alpha,\alpha\beta)=\Mt_{\ell}(G,\alpha,\alpha\beta) \\
  =& \{q\in Q|(\exists\text{ run }q_0\xrightarrow[]{s_1} q_1 \xrightarrow[]{s_2} q_2)[q_0\in Q_0\wedge \ell(s_1)=\alpha \wedge \ell(s_1s_2)=\alpha\beta]\},
\end{align}
\end{subequations}
which is the set of states of $G$ may be in when $\alpha$ has just been generated, given that 
the current observation is $\alpha\beta$.

In the sequel, we sometimes refer to an LFSA $\Scal$ as an FSA $G$ with respect to a labeling function
$\ell$.




\subsection{The concurrent composition}

We now recall our concurrent composition that will be used as a main tool in this paper.
The concurrent composition was first proposed in 
\cite{Zhang2020DetPNFA} to verify the \emph{negation} of strong detectability in 
arbitrary LFSAs and later on used in \cite{Zhang2023PTIMEVerEnforDetDES,Zhang2023UnifiedFrame4DES,Miao2025StrongInitialFinalOpacity}, for the verification of a series of properties such as strong detectability, diagnosability, and opacity. 
The concurrent-composition method provides a unified and currently the most efficient method to verifying 
all kinds of inference-based properties such as strong detectability, diagnosability, and predictability in LFSAs
\cite{Zhang2023UnifiedFrame4DES}, while the classical widely-used methods --- the detector
method \cite{Shu2011GDetectabilityDES} for verifying strong detectability, the twin-plant method
\cite{Jiang2001PolyAlgorithmDiagnosabilityDES} and the verifier method 
\cite{Yoo2002DiagnosabiliyDESPTime,Genc2009PredictabilityDES} for verifying diagnosability and
predictability are not universal and
all depend on two fundamental assumptions: deadlock-freeness (there are no dead states) and
divergence-freeness (no infinite runs of unobservable events can be generated), because these old methods were
proposed to verify inference-based properties themselves.
Apart from verifying inference-based properties, when used in combination with an
observer\footnote{Originally
proposed in \cite{RabinScott1959PowersetConstruction} for determinizing nondeterministic finite
automata with $\varepsilon$-transitions. The terminology ``observer'' dates back to 
\cite{Ozveren1990ObservabilityDES,Shu2007Detectability_DES}.}, the concurrent composition provides
a unified and currently the most efficient method to verifying all kinds 
of standard versions of state-based opacity and strong versions of state-based opacity in LFSAs
\cite{Zhang2023UnifiedFrame4DES,Han2023StrongCSOpacity_DES}.
See \cite[Table~4 and Table~5]{Zhang2023UnifiedFrame4DES} for details. In summary, the concurrent-composition 
method has fundamentally changed the state-of-the-art of discrete-event systems modeled by LFSAs.

\begin{definition}[\cite{Zhang2020DetPNFA}]\label{FA:def_CCa}
	Consider two LFSAs $\Scal^i=(Q_i,E_i,\dt_i,Q_{0i},\Sig,\ell_i)$, $i=1,2$, where $E_{io}$ and $E_{iuo}$ denote 
	the set of observable events of $\Scal^i$ and the set of unobservable events of $\Scal^i$, respectively.
	The \emph{concurrent
	composition}\footnote{Note that the concurrent composition can be naturally extended to extended LFSAs 
	in the sense that the labeling functions are extended to $\delta\to\Sig\cup\{\ep\}$, that is, each observable
  transition is assigned to a label of $\Sig$ and each unobservable transition is assigned to $\ep$.}
	$\CC(\Scal^1,\Scal^2)$ 
	of $\Scal^1$ and $\Scal^2$ 
	is defined by LFSA
	\begin{equation}\label{FA:eqn8}
	  \CC(\Scal^1,\Scal^2) = (Q',E',\dt',Q_0',\Sig',\ell'),
	\end{equation} where
	\begin{enumerate}
	\item $Q'=Q_1\times Q_2$;
	\item $E'=E_o'\cup E_{uo}'$, where $E_o'=\{(e_1,e_2)|e_1\in E_{1o},e_2\in E_{2o},
		\ell_1(e_1)=\ell_2(e_2)\}$,
		$E_{uo}'=\{(e_1,\epsilon)|e_1\in E_{1uo}\}\cup
		\{(\epsilon,e_2)|e_2\in E_{2uo}\}$;
	  \item for all $(q_1,q_2),(q_3,q_4)\in Q'$, $(e_{1o},e_{2o})
		\in E_o'$, $(e_{1uo},\epsilon),(\epsilon,e_{2uo})\in E_{uo}'$,
		\begin{itemize}
		  \item $((q_1,q_2),(e_{1o},e_{2o}),(q_3,q_4))\in\dt'$ 
			if and only if $(q_1,e_{1o},q_3)\in\dt_1$, $(q_2,e_{2o},q_4)\in\dt_2$,
		  \item $((q_1,q_2),(e_{1uo},\epsilon),(q_3,q_4))\in\dt'$ 
			if and only if $(q_1,e_{1uo},q_3)\in\dt_1$, $q_2=q_4$,
		  \item $((q_1,q_2),(\epsilon,e_{2uo}),(q_3,q_4))\in\dt'$ 
			if and only if $q_1=q_3$, $(q_2,e_{2uo},q_4)\in\dt_2$;
		\end{itemize}
	\item $Q_0'=Q_{01}\times Q_{02}$;
	\item 
	  for all $(e_{1o},e_{2o})\in E_o'$, $(e_{1uo},\epsilon)\in E_{uo}'$, and
	  $(\epsilon,e_{2uo})\in E_{uo}'$, $\ell'((e_{1o},e_{2o})):=\ell_1(e_{1o})=\ell_2(e_{2o})$,
	  $\ell'((e_{1uo},\epsilon)):=\ell_1(e_{1uo})=\ep$, $\ell'((\epsilon,e_{2uo})):=\ell_2(e_{2uo})=\ep$.
	\end{enumerate}
\end{definition}

The concurrent composition $\CC(\Scal^1,\Scal^2)$ exactly aggregates every pair of 
a run of $\Scal^1$ and a run of $\Scal^2$ that generate the same label sequence.
Particularly, if for $\Scal^1$ and $\Scal^2$, it holds that $E_1=E_2$ and $\ell_1
=\ell_2$, then denote $\CC(\Scal^1,\Scal^2) = \CC_{\ell_1}(G^1,G^2)$, where $G^i=
(Q_i,E_i,\dt_i,Q_{0i})$, $i=1,2$.

See \autoref{FA:exam19} for an illustration.

\begin{example}\label{FA:exam19} 
  An LFSA $\Scal$ and its self-composition $\CC(\Scal)$ are shown in \autoref{FA:fig26}.
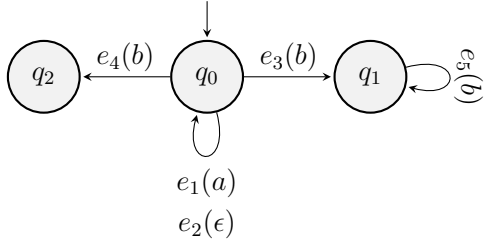
	\begin{figure}[H]
     \centering
	 \subcaptionbox*{$\Scal$, where $q_0$ is the unique initial state which is indicated by an input arrow from nowhere, the symbol in a bracket after an event denotes the label of the event, e.g., $\ell(e_1)=a$, $\ell(e_2)=\ep$.}{
\begin{tikzpicture}[>=stealth',shorten >=1pt,auto,node distance=2.4 cm, scale = 0.9, transform shape,
	>=stealth,inner sep=2pt]

	\node[initial, initial where =above, state] (s0) {$q_0$};
	\node[state] (s1) [right of =s0] {$q_1$};
	\node[state] (s2) [left of =s0] {$q_2$};
	
	\path [->]
	(s0) edge [loop below] node [below, sloped] {$\begin{matrix}e_1(a)\\e_2(\epsilon)\end{matrix}$} (s0)
	(s0) edge node [above, sloped] {$e_3(b)$} (s1)
	(s0) edge node [above, sloped] {$e_4(b)$} (s2)
	(s1) edge [loop right] node [above, sloped] {$e_5(b)$} (s1)
	;

        \end{tikzpicture}
		}
		\hspace{0.3cm}
        \centering
		\subcaptionbox*{Reachable part of $\CC(\Scal)$.}{ 
\begin{tikzpicture}[>=stealth',shorten >=1pt,auto,node distance=3.0 cm, scale = 0.9, transform shape,
	>=stealth,inner sep=2pt]
		\node[rectangular state] (s1s2) {$q_1,q_2$};
		\node[rectangular state,initial,initial where=above] (s0s0) [right of =s1s2] {$q_0,q_0$};
		\node[rectangular state] (s2s1) [below of =s0s0] {$q_2,q_1$};
		\node[rectangular state] (s1s1) [right of =s0s0] {$q_1,q_1$};
		\node[rectangular state] (s2s2) [below of =s1s1] {$q_2,q_2$};

		\path [->]
		(s0s0) edge [out = 210, in = 240, loop] node [below, sloped] {$\begin{matrix}(e_1,e_1)\\(e_2,\epsilon)\\(\epsilon,e_2)\end{matrix}$} (s0s0)
		(s0s0) edge node [above, sloped] {$(e_3,e_4)$} (s1s2)
		(s0s0) edge node [above, sloped] {$(e_3,e_3)$} (s1s1)
		(s0s0) edge node [above, sloped] {$(e_4,e_3)$} (s2s1)
		(s0s0) edge node [above, sloped] {$(e_4,e_4)$} (s2s2)
		(s1s1) edge [loop below] node [below, sloped] {$(e_5,e_5)$} (s1s1)
		;
	\end{tikzpicture}
	}
	\caption{}
	\label{FA:fig26} 
	\end{figure} 
\end{example}

\subsection{The observer}

Now we recall the observer which will also be used as a main tool in the paper.

\begin{definition}[\cite{RabinScott1959PowersetConstruction,Shu2007Detectability_DES}]\label{FA:def_obs}
  Consider an LFSA $\Scal=(G,\Sigma,\ell)$, where $G$ is as in \eqref{FSA}.
  Its \emph{observer} $\Obs(\Scal)=\Obs_{\ell}(G)$\footnote{The term ``observer''
  dates back to \cite{Ozveren1990ObservabilityDES,Shu2007Detectability_DES}.} is defined by
  a deterministic FSA
	\begin{align}\label{FA:eqn_observer}
	  (Q_{\obs},\ell(E_o),\dt_{\obs},q_{0\obs}),
	\end{align}
	where
		\begin{enumerate}
			\item $Q_{\obs}=2^Q$,
			\item $\ell(E_o)=\ell(E)\setminus\{\ep\}$,
			\item for all $X\in Q_{\obs}$ and $a\in\ell(E_o)$, $\dt_{\obs}(X,a)= \bigcup_{q\in X}\bigcup_{
			  \substack{e\in E_o\\\ell(e)=a}}\bigcup_{s\in (E_{uo})^*}\dt(q,es)$,
			\item $q_{0\obs}=\bigcup_{q_0\in Q_0}\bigcup_{s\in (E_{uo})^*}\dt(q_0,s)$.
		\end{enumerate}
\end{definition}

The observer $\Obs(\Scal)$, actually the powerset construction 
\cite{RabinScott1959PowersetConstruction},
can be computed in time exponential in the size of $\Scal$, and has been extensively used 
for many years in both the computer science community and the control community.
The observer $\Obs(\Scal)$ aggregates state estimates of $\Scal$ along all generated
label sequences, which is formulated as follows:
\begin{fact}\label{fact1:High-OrderOpacity}
  for every run $q_{0\obs}\xrightarrow[]{\alpha}X$ of $\Obs(\Scal)$, $X=\Mt(\Scal,\alpha)$.
\end{fact}

In order to derive the tool to verify order-$2$ bygone-state opacity, we additionally
need the inverse observer:

\begin{definition}\label{FA:def_inv_obs}
  Consider an LFSA $\Scal=(G,\Sigma,\ell)$, where $G$ is as in \eqref{FSA}.
  Its \emph{inverse observer} $\Obs^{\inv}(\Scal)=\Obs_{\ell}^{\inv}(G)$ is defined by
  the observer $\Obs_{\ell}(G^{\inv})$ of the inverse LFSA $(G^{\inv},\Sig,\ell)$:
	\begin{align}\label{FA:eqn_inv_observer}
	  (Q_{\obs},\ell(E_o),\dt_{\obs}^{\inv},Q),
	\end{align}
	where
		\begin{enumerate}
			\item $Q_{\obs}=2^Q$,
			\item $\ell(E_o)=\ell(E)\setminus\{\ep\}$,
			\item for all $X\in Q_{\obs}$ and $a\in\ell(E_o)$, $\dt_{\obs}^{\inv}(X,a) = \{q\in Q|(\exists\text{ run }q\xrightarrow[]{se} q')[(q'\in X) \wedge (\ell(s)=\ep \wedge e\in E \wedge \ell(e) = a)]\}$,
			\item $Q$ is the unique initial state.
		\end{enumerate}
\end{definition}
\begin{fact}\label{fact2:High-OrderOpacity}
  Consider an LFSA $S=(G,\Sig,\ell)$ and $\alpha,\beta\in\Sig^*$, $\Mt(\Scal,\alpha,\alpha\beta)=X_1\cap X_2$, where $X_1=\dt_{\obs}(q_{0\obs},\alpha)$, $X_2=\dt_{\obs}^{\inv}(Q,\beta^R)$.
\end{fact}

See \autoref{exam12:High-OrderProperty} for an illustration.

\begin{example}\label{exam12:High-OrderProperty} 
  An LSFA $\Scal=(G,\Sig,\ell)$, its observer $\Obs(\Scal)=\Obs_{\ell}(G)$, its inverse
  LFSA $(G^{\inv},\Sig,\ell)$, and its inverse observer $\Obs^{\inv}(\Scal)=\Obs_{\ell}^{\inv}(G)$ are shown in \autoref{FA:fig50}. Then by \autoref{fact1:High-OrderOpacity}
  and \autoref{fact2:High-OrderOpacity}, we have
  \begin{subequations}\label{eqn47:High-OrderProperty}
	\begin{align}
	  \Mt_{\ell}(G,abb) &= \dt_{\obs}(q_{0\obs},abb) = \{q_1\},\label{eqn47-1:High-OrderProperty}\\
	  \Mt_{\ell}(G,abb,abbab) &= \dt_{\obs}(q_{0\obs},abb)\cap \dt_{\obs}^{\inv}(Q,ba)= \{q_1\}\cap \{q_0\} = \emptyset,\label{eqn47-2:High-OrderProperty}\\
	  \Mt_{\ell}(G,abb,abbbb) &= \dt_{\obs}(q_{0\obs},abb)\cap \dt_{\obs}^{\inv}(Q,bb)= \{q_1\}\cap \{q_0,q_1\} = \{q_1\}.\label{eqn47-3:High-OrderProperty}
	\end{align}
  \end{subequations}
  Next we check the above result by definition.
  For $\Scal$, the unique run that produces label sequence $abb$ is 
  \begin{align*}
	q_0 \xrightarrow[]{e_1} q_0 \xrightarrow[]{e_3} q_1 \xrightarrow[]{e_5} q_1, 
  \end{align*}
  we then obtain \eqref{eqn47-1:High-OrderProperty}.

  No run produces label sequence $abbab$, so \eqref{eqn47-2:High-OrderProperty} holds.

  The unique run that produces label sequence $abbbb$ is 
  \begin{align*}
	q_0 \xrightarrow[]{e_1} q_0 \xrightarrow[]{e_3} q_1 \xrightarrow[]{e_5} q_1 \xrightarrow[]{e_5} q_1 \xrightarrow[]{e_5} q_1. 
  \end{align*}
  Also by $\ell(e_1e_3e_5)=abb$ and $\ell(e_1e_3e_5e_5e_5)=abbbb$, we obtain \eqref{eqn47-3:High-OrderProperty}.

	\begin{figure}[!htbp]
		\centering
		\subcaptionbox{LFSA $\Scal=(G,\Sig,\ell)$.\label{FA:fig26'}}{
			\begin{tikzpicture}[>=stealth',shorten >=1pt,auto,node distance=2.2 cm, scale = 0.9, transform shape,
	>=stealth,inner sep=2pt]

	\node[initial, initial where =above, state] (s0) {$q_0$};
	\node[state] (s1) [right of =s0] {$q_1$};
	\node[state] (s2) [left of =s0] {$q_2$};
	\node[state] (s3) [below =0.5cm of s2] {$q_3$};
	
	\path [->]
	(s0) edge [loop below] node [below, sloped] {$e_1(a)$} (s0)
	(s0) edge node [above, sloped] {$e_3(b)$} (s1)
	(s0) edge node [above, sloped] {$e_4(b)$} (s2)
	(s1) edge [loop right] node [above, sloped] {$e_5(b)$} (s1)
	(s0) edge node [above, sloped] {$e_2(\epsilon)$} (s3)
	;

        \end{tikzpicture}
	  }\hspace{0.5cm}
	  \subcaptionbox{Reachable part of observer $\Obs(\Scal)=\Obs_{\ell}(G)$.\label{FA:fig27}}{
	\begin{tikzpicture}[>=stealth',shorten >=1pt,auto,node distance=2.2 cm, scale = 0.9, transform shape,>=stealth,inner sep=2pt]

	  \node[rectangular state,initial,initial where=left] (s0) {$\{q_0,q_3\}$};
	  \node[rectangular state] (s1s2) [right of =s0] {$\{q_1,q_2\}$};
	  \node[rectangular state] (s1) [right of =s1s2] {$\{q_1\}$};
	  \node[rectangular state] (emptys) [right of =s1] {$\emptyset$};

	  \path [->]
	  (s0) edge [loop above] node {$a$} (s0)
	  (s0) edge node [above, sloped] {$b$} (s1s2)
	  ;
	  \path [->]
	  (s1s2) edge node [above, sloped] {$b$} (s1)
	  (s1s2) edge [bend right] node [below, sloped] {$a$} (emptys)
	  ;
	  \path [->]
	  (s1) edge [loop above] node {$b$} (s1)
	  (s1) edge node {$a$} (emptys)
	  (emptys) edge [loop above] node {$a,b$} (emptys)
	  ;
	\end{tikzpicture}
    }

	\subcaptionbox{LFSA $(G^{\inv},\Sig,\ell)$.\label{FA:fig51}}{
			\begin{tikzpicture}[>=stealth',shorten >=1pt,auto,node distance=2.2 cm, scale = 0.9, transform shape,
	>=stealth,inner sep=2pt]

	\node[initial, initial where =above, state] (s0) {$q_0$};
	\node[state] (s1) [right of =s0] {$q_1$};
	\node[state] (s2) [left of =s0] {$q_2$};
	\node[state] (s3) [below =0.5cm of s2] {$q_3$};
	
	\path [->]
	(s0) edge [loop below] node [below, sloped] {$e_1(a)$} (s0)
	(s1) edge node [above, sloped] {$e_3(b)$} (s0)
	(s2) edge node [above, sloped] {$e_4(b)$} (s0)
	(s1) edge [loop right] node [above, sloped] {$e_5(b)$} (s1)
	(s3) edge node [above, sloped] {$e_2(\epsilon)$} (s0)
	;

        \end{tikzpicture}
	  }\hspace{0.5cm}
	  \subcaptionbox{Part of inverse observer $\Obs^{\inv}(\Scal)=\Obs_{\ell}^{\inv}(G)$.\label{FA:fig52}}{
		\begin{tikzpicture}[>=stealth',shorten >=1pt,auto,node distance=2.2 cm, scale = 0.9, transform shape,>=stealth,inner sep=2pt, align=center]

		  \node[rectangular state,initial,initial where=left] (0123) {$\{q_0,q_1,$\\$q_2,q_3\}$};
		  \node[rectangular state] (0) [below left of =0123] {$\{q_0\}$};
		  \node[rectangular state] (01) [below right of =0123] {$\{q_0,q_1\}$};

	  \path [->]
	  (0123) edge node {$a$} (0)
	  (0123) edge node {$b$} (01)
	  (01) edge node {$a$} (0)
	  (0) edge [loop left] node {$a$} (0)
	  (01) edge [loop right] node {$b$} (01)
	  ;
	\end{tikzpicture}
    }
	\caption{} 
	\label{FA:fig50}
	\end{figure} 
\end{example}

\begin{remark}\label{rem5:High-OrderProperty}
  The above checks by definition as in \autoref{exam12:High-OrderProperty} generally cannot be done, because for an arbitrary given LFSA $(G,\Sig,\ell)$, for a label sequence $\alpha\beta$, there may exist infinitely many runs that produce $\alpha\beta$.
\end{remark}

With these preliminaries, we are ready to show the main results.

\section{Order-\texorpdfstring{$2$}{2} bygone-state opacity}
\label{sec:Order2-ByGState-Ooacity}

\subsection{Formulation}

Consider an FSA $G=(Q,E,\dt,Q_0)$, two agents $A_1$ and $A_2$ with observable event sets $E_1\subset E$
and $E_2\subset E$ and labeling functions $\ell_1$ and $\ell_2$, where $\ell_i: E_i\to\Sig_i,
E\setminus E_i\to\{\ep\}$ with $\Sig_i$ an alphabet, $i=1,2$.
As mentioned before, assume both agents know the structure of $G$,
also assume $A_2$ knows $E_1$ and $\ell_1$ but cannot observe events of $E_1\setminus E_2$.
Denote
\begin{subequations}\label{eqn37:High-OrderOpacity} 
  \begin{align}
	(G,\Sig_i,\ell_i) &=: G_{A_i}, & (G^{\inv},\Sig_i,\ell_i) &=: G_{A_i}^{\inv},\\
	\Mt_{\ell_i}(G,\alpha) &=: \Mt_i(G,\alpha), & \Mt_{\ell_i}(G,\alpha,\alpha\beta) &=: \Mt_i(G,\alpha,\alpha\beta), \\
	\Obs_{\ell_i}(G) &=: \Obs_{A_i}, & \Obs_{\ell_i}^{\inv}(G) &=: \Obs_{A_i}^{\inv},
  \end{align}
\end{subequations}
for short, $i=1,2$.
We formulate the definition of order-$2$ bygone-state opacity as follows.

\begin{definition}[Ord2BGSO]\label{def:order2-ByGState-Ooacity}
  An FSA $G=(Q,E,\dt,Q_0)$ is called \emph{order-2 bygone-state opaque with respect to 
  agents $A_1$ and $A_2$} if for every run 
  \begin{align}\label{eqn41:High-OrderProperty}
	q_{01} \xrightarrow[]{s_1} q_1 \xrightarrow[]{s_2} q_2 =: \pi_1
  \end{align}
  with $q_{01}\in Q_0$, if $|\Mt_1(G,\ell_1(s_1),\ell_1(s_1s_2))| = 1$, then
  there exists another run 
  \begin{align}\label{eqn42:High-OrderProperty}
	q_{02} \xrightarrow[]{s_3} q_3 \xrightarrow[]{s_4} q_4 =: \pi_2
  \end{align}
  with $q_{02} \in Q_0$ such that $\ell_2(s_1) = \ell_2 (s_3)$, $\ell_2(s_2) = \ell_2 (s_4)$, 
  and $|\Mt_1(G,\ell_1(s_3),\ell_1(s_3s_4))| > 1$. 
\end{definition}

Intuitively, if automaton $G$ satisfies the order-$2$ bygone-state opacity with respect to agents $A_1$ and $A_2$,
then agent $A_2$ cannot be sure with respect to any (e.g., $\ell_2(s_1s_2)$ in 
\autoref{def:order2-ByGState-Ooacity})
of its observations to $G$, whether $A_1$ can uniquely determine the state of $G$ at the current and any past time.

\subsection{Verification}

In this subsection, we derive a method to verify order-$2$ bygone-state opacity 
for an FSA $G$ with respect to agents $A_1$ and $A_2$. In \autoref{def:order2-ByGState-Ooacity}, the run \eqref{eqn41:High-OrderProperty} is quantified by $\forall$, the run
\eqref{eqn42:High-OrderProperty} is quantified by $\exists$, where the two runs look
the same to agent $A_2$. Therefore, in order to characterize this property, we need to 
synchronize every run of $G$ as in \eqref{eqn41:High-OrderProperty} and a family of runs 
of $G$ as in \eqref{eqn42:High-OrderProperty}, where the family from 
bygone-state estimates of $A_1$ to $G$ according to its own observations to $G$,
as defined in \eqref{eqn:ByGSE}, where the synchronization will be realized by the 
concurrent composition.
We also need to compute for every observation $\alpha\beta$
of $A_2$ to $G$, all possible bygone-state estimates of $A_1$ to $G$ with respect to 
all runs of $G$ that produce $\alpha\beta$ under $\ell_2$. To this end, we need the
\emph{order-$2$ observer} $\Obs_{A_1\leftarrow A_2}(G)$
proposed in \cite{Zhang2024High-OrderPropertyDES}. Note that $\Obs_{A_1\leftarrow A_2}(G)$
can do the above computation when $\beta=\ep$, so, we additionally need to extend
$\Obs_{A_1\leftarrow A_2}(G)$. This extension is nontrivial, because we need to 
construct it from scratch, and the constructing $\Obs_{A_1\leftarrow A_2}(G)$
is only one step in constructing the extension.

Next we construct an extended version of the order-$2$ observer $\Obs_{A_1\leftarrow A_2}(G)$
to verify order-$2$ bygone-state opacity. The construction will be done in a few steps as
follows.

\subsubsection{Step~1}

First, we give a representation for the interleaving of languages generated by two FSAs.

\begin{theorem}\label{thm9:High-OrderCSO}
  Consider two LFSAs $\Scal_i=(G_i,\Sig_i,\ell_i)$, where $G_i=(Q_i,E_i,\dt_i,Q_{0i})$, $i=1,2$. Assume $\ell_i(E_i)=\{\ep\}$, $i=1,2$. Then $L_{\setminus\ep}(\CC(\Scal_1,\Scal_2))
  = L(G_1) ||| L(G_2)$, where $L_{\setminus\ep}(\CC(\Scal_1,\Scal_2))=
  \{e_1\dots e_n| (\exists (e_{1L},e_{1R})\dots(e_{nL},e_{nR}) \in L(\CC(\Scal_1,\Scal_2)))
  [ (\forall 1\le i\le n) [e_{iL}\ne\ep \implies e_i=e_{iL}, e_{iR}\ne\ep \implies
  e_i=e_{iR} ] ]\}$.
\end{theorem}

\begin{proof}
  This result naturally holds by definition.
\end{proof}

By \autoref{thm9:High-OrderCSO} we have the following:
\begin{theorem}\label{thm10:High-OrderCSO}
  Consider an LFSA $\Scal=(G,\Sig,\ell)$, where $G=(Q,E,\dt,Q_0)$ is as in \eqref{FSA}.
  Consider its observer $\Obs_{\ell}(G) = (Q_{\obs},\ell(E_o),\dt_{\obs},q_{0\obs})$
  as in \eqref{FA:eqn_observer} and its inverse observer 
  $\Obs_{\ell}^{\inv}(G) = (Q_{\obs},\ell(E_o),\dt_{\obs}^{\inv},Q)$ as in 
\eqref{FA:eqn_inv_observer}. Then 
$$L\left( \Obs_{\ell}(G) \right) \interleave
  L\left( \Obs_{\ell}^{\inv}(G) \right) =
  L_{\setminus \ep}\left( \CC_{\ell_{\ep}}\left( \Obs_{\ell}(G),\Obs_{\ell}^{\inv}(G) \right) \right).$$
\end{theorem}

Intuitively, \autoref{thm10:High-OrderCSO} means that the interleaving of the languages
generated by observer $\Obs_{\ell}(G)$ and by inverse observer $\Obs_{\ell}^{\inv}(G)$
can be computed by the concurrent composition of $\Obs_{\ell}(G)$ and $\Obs_{\ell}^{\inv}(G))$ after endowing them the $\ep$-labeling function $\ell_{\ep}$.

\subsubsection{Step~2}

Second, we recall the order-2 observer and relevant results proposed in \cite{Zhang2024High-OrderPropertyDES}.

For a label sequence $\alpha$ observed by $A_2$, the real generated event sequence can be 
any $s\in \ell_2^{-1}(\alpha)\cap L(G)$, so the observation of $A_1$ can be $\ell_1(s)$
for any such $s$,
and then the inference of $A_1$'s current-state estimate from $A_2$ can be $\Mt_1(G,\ell_1(s))$
for any such $s$.
Formally, when $A_2$ observes $\alpha\in \ell_2(L(G))$, all possible inferences of $A_1$'s current-state estimate of $G$
by $A_2$ are represented by the set
\begin{subequations}\label{eqn:order-2-CSE}
  \begin{align}
	& \Mt_{A_1\leftarrow A_2}(G,\alpha)\label{eqn15_1:High-OrderOpacity}\\
	:= &\{\Mt_1(G,\ell_1(s))|s\in \ell_2^{-1}(\alpha)\cap L(G)\} \subset 2^Q,\label{eqn15_2:High-OrderOpacity}
  \end{align}
\end{subequations}
which is called \emph{order-$2$ current-state estimate of $G$ with respect to agents $A_1$ and $A_2$}.
By definition, $\Mt_{A_1\leftarrow A_2}(G,\alpha)$ must contain $A_1$'s real current-state estimate of $G$.

\begin{definition}[\cite{Zhang2024High-OrderPropertyDES}]\label{def:order2observer}
  Consider FSA $G$ as in \eqref{FSA}, agents $A_1$ and $A_2$ with observable event sets $E_1\subset E$
  and $E_2\subset E$ and labeling functions $\ell_1$ and $\ell_2$.
  Denote LFSAs $(G,\Sig_i,\ell_i)$ by $G_{A_i}$ for short, $i=1,2$. 
\begin{enumerate}
  \item Compute the observer $\Obs_1(G)$ of $G_{A_1}$, and denote $\Obs_1(G)$ as $\Obs_{A_1}$ for short.
  \item Compute the concurrent composition $\CC(G_{A_1},\Obs_{A_1})$ of LFSA $G_{A_1}$ and its observer $\Obs_{A_1}$,
	replace each event $(e_1,e_2)$ by $e_1$, replace the labeling function of $\CC(G_{A_1},\Obs_{A_1})$
	by $\ell_2$, and denote the modification of $\CC(G_{A_1},\Obs_{A_1})$ by $\CC_{A_1\to A_2}^{G,\Obs}$
	which is an LFSA.
  \item Compute the observer $\Obs\left( \CC_{A_1\to A_2}^{G,\Obs} \right)$ of $\CC_{A_1\to A_2}^{G,\Obs}$, and call
	$\Obs\left( \CC_{A_1\to A_2}^{G,\Obs} \right)$ \emph{order-$2$ observer} and denote it as $\Obs_{A_1\leftarrow A_2}
	(G)$, where $A_1\leftarrow A_2$ intuitively describes the scenario $A_2$ infers $A_1$'s current-state estimate to $G$. 
\end{enumerate}
\end{definition}
See \autoref{fig1:order2SEproperty} for an illustration.
It takes doubly exponential time to compute an order-$2$ observer $\Obs_{A_1\leftarrow A_2}(G)$.

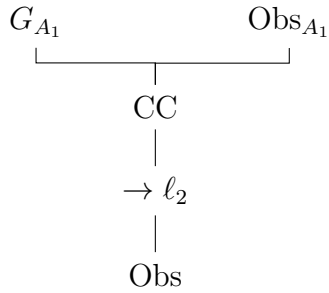
\begin{figure}[!htbp]
  \centering
  \begin{tikzpicture}[circuit logic IEC]
	\matrix[column sep=5mm, row sep=0.5cm]
	{
	  \node (SusrR) {$G_{A_1}$}; & & \node (ObsusrR) {$\Obs_{A_1}$};\\
	  & \node (CCusrR) {$\CC$}; & \\
	  & \node (P1to2R) {$\to \ell_2$}; & \\
	  & \node (ObsIntrR) {$\Obs$}; & \\
	};
	\draw (SusrR.south) -- ++(down:2.0mm) -| (CCusrR.north);
	\draw (ObsusrR.south) -- ++(down:2.0mm) -| (CCusrR.north);
	\draw (CCusrR.south) -- (P1to2R);
	\draw (P1to2R) -- (ObsIntrR.north);
  \end{tikzpicture}
  \caption{Construction of the order-$2$ observer $\Obs_{A_1\leftarrow A_2}(G)$ \cite{Zhang2024High-OrderPropertyDES}.}
  \label{fig1:order2SEproperty}
\end{figure}

The next \autoref{thm7:High-OrderCSO} shows several fundamental properties of the order-$2$ observer
$\Obs_{A_1\leftarrow A_2}(G)$, 
and plays a fundamental role in verifying the order-$2$ estimation-based property
\cite{Zhang2024High-OrderPropertyDES}.

\begin{theorem}[\cite{Zhang2024High-OrderPropertyDES}]\label{thm7:High-OrderCSO}
  Consider an FSA $G$ as in \eqref{FSA}, agents $A_1$ and $A_2$ with observable event sets $E_1\subset E$
  and $E_2\subset E$ and labeling functions $\ell_1$ and $\ell_2$,
  and the order-$2$ observer $\Obs_{A_1\leftarrow A_2}(G)$.
  \begin{enumerate}[(i)]
	\item\label{item1:SEproperty}
	  The initial state of 
	  $\Obs_{A_1\leftarrow A_2}(G)$ is in form of
	  $\{(q_{0,1},X_0),\dots,(q_{0,m},X_0)\}$, where 
	  $\{q_{0,1},\dots,q_{0,m}\}=X_0$, $X_0$ is the initial state of $\Obs_{A_1}$.
	\item\label{item2:SEproperty}
	  $L(G) = L(\CC_{A_1\to A_2}^{G,\Obs})$.
	\item\label{item3:SEproperty}
	  For every reachable state $\{(q_1,X_1),\dots,(q_n,X_n)\}$ of $\Obs_{A_1\leftarrow A_2}(G)$, 
	  $q_j\in X_j$, $j\in\llb1, n\rrb$.
	\item\label{item4:SEproperty}
	  For every run $C_0\xrightarrow[]{\alpha} \{(q_1,X_1),\dots,(q_n,X_n)\}$ of 
	  $\Obs_{A_1\leftarrow A_2}(G)$, where $C_0$ is the initial state, 
	  $\{X_1,\dots,X_n\}=\Mt_{A_1\leftarrow A_2}(G,\alpha)$.
  \end{enumerate}
\end{theorem}

\autoref{thm7:High-OrderCSO} shows that the language generated by the concurrent composition
$\CC_{A_1\to A_2}^{G,\Obs}$ coincides with the language generated by $G$, and
one can use the order-$2$ observer to compute $\Mt_{A_1\leftarrow A_2}(G,\alpha)$.

\subsubsection{Step~3}

In this step, we extend order-$2$ current-state estimate $\Mt_{A_1\leftarrow A_2}(G,\alpha)$ \eqref{eqn:order-2-CSE}
to \emph{order-$2$ bygone-state estimate $\Mt_{A_1\leftarrow A_2}(G,\alpha,\alpha\beta)$}
\eqref{eqn:order-2-ByGSE} just like extending to current-state estimate $\Mt(\Scal,\alpha)$ \eqref{eqn:CSE}
to bygone-state estimate $\Mt(\Scal,\alpha,\alpha\beta)$ \eqref{eqn:ByGSE}.

For a label sequence $\alpha\beta$ observed by agent $A_2$, the real generated event sequence can be 
any $s_1s_2\in \ell_2^{-1}(\alpha\beta)\cap L(G)$ with $\ell_2(s_1)=\alpha$ and $\ell_2(s_1s_2)=\alpha\beta$, so the observation of $A_1$ can be $\ell_1(s_1)\ell_1(s_2)$
for any such $s_1s_2$,
and then the inference of $A_1$'s bygone-state estimate from $A_2$ can be $\Mt_1(G,\ell_1(s_1),\ell_1(s_1s_2))$
for any such $s_1s_2$.
Formally, when $A_2$ observes $\alpha\beta\in \ell_2(L(G))$, all possible inferences of $A_1$'s bygone-state estimates of $G$ by $A_2$ are represented by the set
\begin{subequations}\label{eqn:order-2-ByGSE}
\begin{align}
  & \Mt_{A_1\leftarrow A_2}(G,\alpha,\alpha\beta)\\
  := &\{\Mt_1(G,\ell_1(s_1),\ell_1(s_1s_2))|s_1s_2\in \ell_2^{-1}(\alpha\beta)\cap L(G),
  \ell_2(s_1)=\alpha,\ell_2(s_2)=\beta\} \subset 2^Q,
\end{align}
\end{subequations}
which is called \emph{order-$2$ bygone-state estimate of $G$ with respect to agents $A_1$ and $A_2$ and $\alpha\beta\in\ell(L(G))$}.

By \autoref{fact2:High-OrderOpacity}, we have 
$\Mt(\Scal,\alpha,\alpha\beta)=X\cap Y$, where $X=\dt_{\obs}(q_{0\obs},\alpha)$, $Y=\dt_{\obs}^{\inv}(Q,\beta^R)$.

Denote
\begin{align}\label{eqn43:High-OrderProperty}
\Obs_{A_i} &=: \left( Q_{\obs},\Sig_i,\dt_{\obs}^i,q_{0\obs}^i \right), & \Obs_{A_i}^{\inv} &=: \left( Q_{\obs},\Sig_i,\dt_{\obs}^{i,\inv},Q \right),
\end{align}
where $i=1,2$, recall that $\Obs_{A_i}$ and $\Obs_{A_i}^{\inv}$ are the observer and the inverse observer of LFSA $(G,\Sig_i,\ell_i)$, respectively.

Next, we construct a tool to compute order-$2$ bygone-state estimate $\Mt_{A_1\leftarrow A_2}(G,\alpha,\alpha\beta)$ \eqref{eqn:order-2-ByGSE}.

\begin{enumerate}[(1)]
  \item\label{item1:High-OrderProperty}
	Compute concurrent compositions $\CC_{\ell_1}(G,\Obs_{A_1})$ and $\CC_{\ell_1}(G^{\inv},\Obs_{A_1}^{\inv})$, replace each event $(e_1,e_2)$ of these two concurrent compositions by $e_1$, and then replace the labeling functions of these two concurrent compositions by $\ell_2$, and denote the modified concurrent compositions by $\CC_{A_1\to A_2}^{G,\Obs}$ and $\CC_{A_1\to A_2}^{G^{\inv},\Obs^{\inv}}$, respectively.
  \item\label{item2:High-OrderProperty}
	Compute the observer $\Obs\left( \CC_{A_1\to A_2}^{G,\Obs} \right)$ of $\CC_{A_1\to A_2}^{G,\Obs}$, compute
	the observer $\Obs\left( \CC_{A_1\to A_2}^{G^{\inv},\Obs^{\inv}} \right)$
	of $\CC_{A_1\to A_2}^{G^{\inv},\Obs^{\inv}}$, 
	where recall that $\Obs\left( \CC_{A_1\to A_2}^{G,\Obs} \right)$ is the order-$2$ observer and is denoted as
	$\Obs_{A_1\leftarrow A_2}(G)$
	for short. Denote $\Obs\left( \CC_{A_1\to A_2}^{G^{\inv},\Obs^{\inv}} \right)
	=: \Obs_{A_1\leftarrow A_2}^{\inv} (G)$ for short, and call it \emph{order-$2$ inverse observer}.
\end{enumerate}

Denote
\begin{subequations}\label{eqn44:High-OrderProperty}
  \begin{align}
	\Obs_{A_1\leftarrow A_2} (G)
	&= \left( 2^{Q\times 2^Q}, \Sig_2, \dt_{\Obs,A_1\leftarrow A_2}^{G,\Obs}, q_{0\obs,A_1\leftarrow A_2}^{G,\Obs}\right),\\
	\Obs_{A_1\leftarrow A_2}^{\inv} (G) &= \left( 2^{Q\times 2^Q}, \Sig_2, \dt_{\Obs,A_1\leftarrow A_2}^{G^{\inv},\Obs^{\inv}}, q_{0\obs,A_1\leftarrow A_2}^{G^{\inv},\Obs^{\inv}}\right),
  \end{align}
\end{subequations}
where $q_{0\obs,A_1\leftarrow A_2}^{G,\Obs} = \left\{ \left( q,q_{0\obs}^1 \right) \left| q\in q_{0\obs}^1 \right. \right\}$, 
$q_{0\obs,A_1\leftarrow A_2}^{G^{\inv},\Obs^{\inv}} = \{ (q,Q) | q\in Q\}$,
recall that $q_{0\obs}^1$ is the initial state of observer $\Obs_{A_1}$,
$Q$ is the initial state of inverse observer $\Obs_{A_1}^{\inv}$.

Based on this construction, we have the following result for computing \eqref{eqn:order-2-ByGSE}.

\begin{theorem}\label{thm11:High-OrderProperty}
  Consider an FSA $G$ as in \eqref{FSA}, agents $A_1$ and $A_2$ with observable event sets $E_1\subset E$
  and $E_2\subset E$ and labeling functions $\ell_1$ and $\ell_2$, 
  and consider two observers $\Obs_{A_1\leftarrow A_2} (G)$ and
  $\Obs_{A_1\leftarrow A_2}^{\inv} (G)$.
  Consider label sequence $\alpha\beta\in\Sig_2^*\cap\ell_2(L(G))$ observed by agent $A_2$. Then
  \begin{subequations}\label{eqn45:High-OrderProperty}
  \begin{align}
	&\Mt_{A_1\leftarrow A_2}(G,\alpha,\alpha\beta)\\
	=& \{X_1\cap X_2|(\exists q\in Q)\\
	  &\left[(q,X_1) \in \dt_{\Obs,A_1\leftarrow A_2}^{G,\Obs} \left( q_{0\obs,A_1\leftarrow A_2}^{G,\Obs},\alpha \right)  \wedge\right.\\
	  & \left.\left. (q,X_2) \in \dt_{\Obs,A_1\leftarrow A_2}^{G^{\inv},\Obs^{\inv}} \left( q_{0\obs,A_1\leftarrow A_2}^{G^{\inv},\Obs^{\inv}}, \beta^R \right) \right]\right\}.
  \end{align}
  \end{subequations}
\end{theorem}

\begin{proof}
  By definition, for a subset $X\subset Q$, $X\in \Mt_{A_1\leftarrow A_2}(G,\alpha,\alpha\beta)$ if and only if
  there exists $s_1s_2\in L(G)$ such that $\ell_2(s_1)=\alpha$, $\ell_2(s_1s_2) = \alpha\beta$, and 
  $X=\Mt_1(G,\ell_1(s_1),\ell_1(s_1s_2))$. By \autoref{fact2:High-OrderOpacity}, $X=X_1\cap X_2$, where
  $X_1=\dt_{\obs}^1(q_{0\obs}^1,\alpha)$, $X_2=\dt_{\obs}^{1,\inv}(Q,\beta^R)$. Then $X\in \Mt_{A_1\leftarrow A_2}
  (G,\alpha,\alpha\beta)$ if and only if there exists a run 
  \begin{align}
	q_0 \xrightarrow[]{s_1} q_1 \xrightarrow[]{s_2} q_2 
  \end{align}
  in $G$ with $q_0 \in Q_0$ such that $\ell_2(s_1)=\alpha$, $\ell_2(s_1s_2) = \alpha\beta$, and
  $X = X_1 \cap X_2$, where $X_1$ and $X_2$ are as above. This also implies $q_1\in X$. All such pairs $(q_1,X_1)$
  can be computed by order-$2$ observer $\Obs_{A_1\leftarrow A_2} (G)$ via 
  $\dt_{\Obs,A_1\leftarrow A_2}^{G,\Obs} (q_{0\obs,A_1\leftarrow A_2}^{G,\Obs},\alpha)$, 
  all such pairs $(q_1,X_2)$ can be computed by order-$2$ inverse observer $\Obs_{A_1\leftarrow A_2}^{\inv} (G)$ via
  $\dt_{\Obs,A_1\leftarrow A_2}^{G^{\inv},\Obs^{\inv}} (q_{0\obs,A_1\leftarrow A_2}^{G^{\inv},\Obs^{\inv}}, \beta^R)$.
  Then, \eqref{eqn45:High-OrderProperty} holds.
\end{proof}

\begin{remark}\label{rem2:High-OrderOpacity}
  In \autoref{thm11:High-OrderProperty}, \eqref{eqn45:High-OrderProperty} provides a doubly-exponential-time algorithm
  for computing order-$2$ bygone-state estimate $\Mt_{A_1\leftarrow A_2}(G,\alpha,\alpha\beta)$.
\end{remark}

\subsubsection{Step~4}

In this step, we derive a tool to verify order-$2$ bygone-state opacity. The tool will be constructed based 
on observers $\Obs_{A_1\leftarrow A_2} (G)$ and $\Obs_{A_1\leftarrow A_2}^{\inv} (G)$,
\autoref{thm10:High-OrderCSO},
and \autoref{thm11:High-OrderProperty}.

\begin{enumerate}[(1)]
  \setcounter{enumi}{2}
  \item\label{item3:High-OrderProperty}
	Compute concurrent composition $\CC_{\ell_{\ep}}\left(\Obs_{A_1\leftarrow A_2} (G),
	\Obs_{A_1\leftarrow A_2}^{\inv} (G) \right)$, where recall that $\ell_{\ep}$ is the 
	$\ep$-labeling function. (Then the state of the newly computed concurrent composition is in form of
	$(Y_1,Y_2)$, where $Y_1,Y_2\subset Q\times 2^Q$.) Change each state $(Y_1,Y_2)$ of the concurrent composition 
	to $Y_1\Cap Y_2$ (see \autoref{subsec:notation}). Denote
	the modified concurrent composition by 
	\begin{subequations}\label{eqn46:High-OrderProperty}
	  \begin{align}
		&\overline{\CC}_{\ell_{\ep}}\left(\Obs_{A_1\leftarrow A_2} (G), \Obs_{A_1\leftarrow A_2}^{\inv} (G) \right)\\
		=: &\left( \Ysf, \bar\Sig_2,\dt_{\Ysf},\ysf_0 ,\{\ep\},\ell_{\ep}\right),
	  \end{align}
	\end{subequations}
	where $\bar\Sig_2 = (\Sig_2\times\{\ep\}) \cup (\{\ep\}\times\Sig_2)$, $\ysf_0 = \left(q_{0\obs,A_1\leftarrow A_2}^{G,\Obs}, q_{0\obs,A_1\leftarrow A_2}^{G^{\inv},\Obs^{\inv}} \right)$.
\end{enumerate}

By definition, one sees that for every $Y\in\Ysf$ and every $\sigma_{\Ysf} \in
\bar\Sig_2$, $|\dt_{\Ysf} (Y,\sigma_{\Ysf})| \le 1$. 
Then in the sequel, we will use the form of a partial function 
$\dt_{\Ysf}: \Ysf\times \bar\Sig_2 \to \Ysf$ instead of its
set-valued function form as in \autoref{FA:def_CCa}.

\begin{remark}\label{rem3:High-OrderProperty}
  It takes doubly exponential time to compute concurrent composition \eqref{eqn46:High-OrderProperty}.
\end{remark}

\begin{theorem}\label{thm12:High-OrderProperty}
  Consider an FSA $G$ as in \eqref{FSA}, agents $A_1$ and $A_2$ with observable event sets $E_1\subset E$
  and $E_2\subset E$ and labeling functions $\ell_1$ and $\ell_2$, and consider the concurrent composition
  $$\overline{\CC}_{\ell_{\ep}}\left(\Obs_{A_1\leftarrow A_2} (G), \Obs_{A_1\leftarrow A_2}^{\inv} (G) \right).$$
  The following hold.
  \begin{enumerate}[(i)]
	\item\label{item7:High-OrderProperty}  
	  For every two event sequences $s_{\Ysf}^1,s_{\Ysf}^2\in (\bar\Sig_2)^*$, if
	  $s_{\Ysf}^1(L) = s_{\Ysf}^2(L)$ and $s_{\Ysf}^1(R) = s_{\Ysf}^2(R)$, 
	  then $\dt_{\Ysf} (\ysf_0, s_{\Ysf}^1) = \dt_{\Ysf} (\ysf_0, s_{\Ysf}^1)$.
	\item\label{item4:High-OrderProperty}
	  For every $s_{\Ysf} \in (\bar\Sig_2)^*$, 
	  $\dt_{\Ysf} (\ysf_0,s_{\Ysf}) = \Mt_{A_1\leftarrow A_2}(G,\alpha,\alpha\beta)$, where $\alpha=s_{\Ycal}(L)$, $\beta=s_{\Ycal}(R)^R$.
	\item\label{item5:High-OrderProperty}
	  $L_{\setminus\ep} \left( \overline{\CC}_{\ell_{\ep}}\left(\Obs_{A_1\leftarrow A_2} (G), \Obs_{A_1\leftarrow A_2}^{\inv} (G) \right) \right) = L\left( \Obs_{A_1\leftarrow A_2} (G) \right) \interleave L\left( \Obs_{A_1\leftarrow A_2}^{\inv} (G) \right) $.
	\item\label{item6:High-OrderProperty}
	  The set of the reachable states of concurrent composition $\overline{\CC}_{\ell_{\ep}}\left(\Obs_{A_1\leftarrow A_2} (G), \Obs_{A_1\leftarrow A_2}^{\inv} (G) \right)$ coincides with the set 
	  $\{\Mt_{A_1\leftarrow A_2}(G,\alpha,\alpha\beta)|\alpha\beta\in\Sig_2^*\}$
	  of order-$2$ bygone-state estimates.
  \end{enumerate}
\end{theorem}

\begin{proof}
  \eqref{item7:High-OrderProperty}: This holds by definition of concurrent composition and because the two observers $\Obs_{A_1\leftarrow A_2} (G)$ and $\Obs_{A_1\leftarrow A_2}^{\inv}$ are both deterministic automata.

  \eqref{item4:High-OrderProperty}:
  By item~\eqref{item7:High-OrderProperty} and \autoref{thm11:High-OrderProperty}, item~\eqref{item4:High-OrderProperty} holds.

  \eqref{item5:High-OrderProperty}:
  By \autoref{thm9:High-OrderCSO}, one has $$L_{\setminus\ep} \left( \CC_{\ell_{\ep}}\left(\Obs_{A_1\leftarrow A_2} (G), \Obs_{A_1\leftarrow A_2}^{\inv} (G) \right) \right) = L\left( \Obs_{A_1\leftarrow A_2} (G) \right) \interleave L\left( \Obs_{A_1\leftarrow A_2}^{\inv} (G) \right).$$ Also because $$L_{\setminus\ep} \left( \overline{\CC}_{\ell_{\ep}}\left(\Obs_{A_1\leftarrow A_2} (G), \Obs_{A_1\leftarrow A_2}^{\inv} (G) \right) \right) = L_{\setminus\ep} \left( \CC_{\ell_{\ep}}\left(\Obs_{A_1\leftarrow A_2} (G), \Obs_{A_1\leftarrow A_2}^{\inv} (G) \right) \right),$$
  item~\eqref{item5:High-OrderProperty} holds.

  \eqref{item6:High-OrderProperty}:
  Item~\eqref{item4:High-OrderProperty} and item~\eqref{item5:High-OrderProperty} imply 
  item~\eqref{item6:High-OrderProperty}.
\end{proof}

Item~\eqref{item7:High-OrderProperty} tells that although there may exist multiple event sequences such that their left components (e.g., $s_{\Ysf}^1(L)$) are the same and their right components (e.g., $s_{\Ysf}^1(R)$) are the same, these event sequences lead the concurrent composition $\overline{\CC}_{\ell_{\ep}}\left(\Obs_{A_1\leftarrow A_2} (G), \Obs_{A_1\leftarrow A_2}^{\inv} (G) \right)$ to the same state which is exactly the order-$2$ bygone-state estimate $\Mt_{A_1\leftarrow A_2}(G,s_{\Ysf}^1(L),s_{\Ysf}^1(L)s_{\Ysf}^1(R)^R)$.

\autoref{thm12:High-OrderProperty} implies that the concurrent composition 
$\overline{\CC}_{\ell_{\ep}}\left(\Obs_{A_1\leftarrow A_2} (G), \Obs_{A_1\leftarrow A_2}^{\inv} (G) \right)$
can compute order-$2$ bygone-state estimates, and can characterize all order-$2$ bygone-state estimates.
Based on these conclusions, $\overline{\CC}_{\ell_{\ep}}\left(\Obs_{A_1\leftarrow A_2} (G), \Obs_{A_1\leftarrow A_2}^{\inv} (G) \right)$ 
can verify order-$2$ bygone-state opacity.

\begin{theorem}\label{thm13:High-OrderProperty}
  Consider an FSA $G$ as in \eqref{FSA}, agents $A_1$ and $A_2$ with observable event sets $E_1\subset E$
  and $E_2\subset E$ and labeling functions $\ell_1$ and $\ell_2$, and consider the concurrent composition
  $$\overline{\CC}_{\ell_{\ep}}\left(\Obs_{A_1\leftarrow A_2} (G),\Obs_{A_1\leftarrow A_2}^{\inv} (G) \right)$$
  as in \eqref{eqn46:High-OrderProperty}.
  FSA $G$ is order-$2$ bygone-state opaque with respect to agents $A_1$ and $A_2$ if and only if
  for every reachable
  state $Y$ of concurrent composition \eqref{eqn46:High-OrderProperty}, if there is $X_1\in Y$ with $|X_1|=1$,
  then there exists $X_2\in Y$ such that $|X_2|>1$.
\end{theorem}

\begin{proof}
  By definition, $G$ is order-$2$ bygone-state opaque with respect to agents $A_1$ and $A_2$ if and only if for every run 
  \begin{align*}
	q_{01} \xrightarrow[]{s_1} q_1 \xrightarrow[]{s_2} q_2 =: \pi_1
  \end{align*}
  with $q_{01}\in Q_0$ satisfying $|\Mt_1(G,\ell_1(s_1),\ell_1(s_1s_2))| = 1$, 
  there exists another run 
  \begin{align*}
	q_{02} \xrightarrow[]{s_3} q_3 \xrightarrow[]{s_4} q_4 =: \pi_2
  \end{align*}
  with $q_{02} \in Q_0$ such that $\ell_2(s_1) = \ell_2 (s_3)$, $\ell_2(s_2) = \ell_2 (s_4)$, 
  and $|\Mt_1(G,\ell_1(s_3),\ell_1(s_3s_4))| > 1$.

  Here $s_1$ and $s_3$ are observationally equivalent to $A_2$ ($\ell_2$), $s_2$ and $s_4$ 
  are also observationally equivalent to $A_2$. Then both $\Mt_1(G,\ell_1(s_1),\ell_1(s_1s_2))$ 
  and $\Mt_1(G,\ell_1(s_3),\ell_1(s_3s_4))$ are inferences of $A_1$'s bygone-state estimates to 
  $G$ done by $A_2$ according to $A_2$'s observations $\ell_2(s_1)$ and $\ell_2(s_1s_2)$ to $G$.
  By \autoref{thm12:High-OrderProperty}, $A_2$'s inferences of $A_1$'s bygone-state estimates to 
  $G$ based on $A_2$'s observations to $G$ and the structure of $G$ coincide with the
  reachable states of concurrent composition 
  $\overline{\CC}_{\ell_{\ep}}\left(\Obs_{A_1\leftarrow A_2} (G), \Obs_{A_1\leftarrow A_2}^{\inv} (G) \right)$ as in \eqref{eqn46:High-OrderProperty}. Then the above necessary and sufficient
  condition for order-$2$ bygone-state opacity can be equivalently restated as
  for every reachable
  state $Y$ of concurrent composition \eqref{eqn46:High-OrderProperty}, if there is $X_1\in Y$ 
  (corresponding to $\Mt_1(G,\ell_1(s_1),\ell_1(s_1s_2))$) with $|X_1|=1$,
  then there exists $X_2\in Y$ (corresponding to $\Mt_1(G,\ell_1(s_3),\ell_1(s_3s_4))$) 
  such that $|X_2|>1$.
\end{proof}

\begin{remark}\label{rem4:High-OrderProperty}
  By \autoref{rem3:High-OrderProperty}, \autoref{thm13:High-OrderProperty} provides a
  doubly-exponential-time algorithm for verifying order-$2$ bygone-state opacity of FSA $G$
  with respect to agents $A_1$ and $A_2$.
\end{remark}

See \autoref{exam11:High-OrderProperty} for an illustration.

\begin{example}\label{exam11:High-OrderProperty}
  Consider FSA $G_{ \label{cou1}}$ as in \autoref{fig1:High-OrderCSO}.
  We consider two agents $A_1$ and $A_2$ with observable event sets $E_1=\{b,c,d\}$ and $E_2=\{a,b\}$, respectively, their labeling functions are such that $\ell_1(b) = \ell_1(c) = \alpha$, $\ell_1(d) = \beta$, $\ell_2(a) = \gamma$, $\ell_2(b) = \zeta$.
  We use \autoref{thm13:High-OrderProperty} to verify order-$2$ begone-state opacity of
  $G$ with respect to $A_1$ and $A_2$.

  Observer $\Obs_{A_1}$, concurrent composition $\CC(G_{A_1},\Obs_{A_1})$, and 
  order-$2$ observer $\Obs_{A_1\leftarrow A_2}(G)$ are shown in \autoref{fig2:High-OrderCSO},
  \autoref{fig3:High-OrderCSO}, and \autoref{fig4:High-OrderCSO}, respectively.

  Inverse FSA $G^{\inv}$, inverse observer $\Obs_{A_1}^{\inv}$, concurrent composition $\CC\left( G_{A_1}^{\inv},\Obs_{A_1}^{\inv} \right)$, order-$2$ inverse observer
  $\Obs_{A_1\leftarrow A_2}^{\inv}(G)$, and concurrent composition 
  $\CC_{\ell_{\ep}} \left( \Obs_{A_1\leftarrow A_2}(G), \Obs_{A_1\leftarrow A_2}^{\inv}(G) \right)$
  are shown in \autoref{fig1:High-OrderProperty}, \autoref{fig2:High-OrderProperty},
  \autoref{fig3:High-OrderProperty}, \autoref{fig4:High-OrderProperty},
  and \autoref{fig5:High-OrderProperty},
  respectively.

  We illustrate \autoref{thm12:High-OrderProperty} and \autoref{thm13:High-OrderProperty}.
  In \autoref{fig5:High-OrderProperty}, we choose two event sequences $(\gamma,\ep)(\zeta,\ep)\\(\ep,\zeta)=:s_{\Ysf}^1$ 
  and $(\gamma,\ep)(\ep,\zeta)(\zeta,\ep)=:s_{\Ysf}^2$. One sees $s_{\Ysf}^1(L)=s_{\Ysf}^2(L)=\gamma\zeta$ and $s_{\Ysf}^1(R)=s_{\Ysf}^2(R)=\zeta$.
  These two sequences lead the
  concurrent composition $\CC_{\ell_{\ep}} \left( \Obs_{A_1\leftarrow A_2}(G),
  \Obs_{A_1\leftarrow A_2}^{\inv}(G) \right)$ to the same state $(Y_2,Y_2^{\inv})$.
  This illustrates item~\eqref{item7:High-OrderProperty} of \autoref{thm12:High-OrderProperty}.
  $Y_2\Cap Y_2^{\inv}=\{B\cap F\}=\{\{2\}\}$. Then item~\eqref{item4:High-OrderProperty} of
  \autoref{thm12:High-OrderProperty} implies that $\Mt_{A_1\leftarrow A_2}(G,\gamma\zeta,\gamma\zeta\zeta)=Y_2\Cap Y_2^{\inv}=\{\{2\}\}$.
  Next, we directly compute $\Mt_{A_1\leftarrow A_2}(G,\gamma\zeta,\gamma\zeta\zeta)$
  by definition as follows. (This is possible because $\ell_2^{-1}(\gamma\zeta\zeta)$ is 
  finite. However, generally it is not the case, and so it is not possible to directly
  compute an order-$2$ current-state estimate by definition.)
  $\ell_2^{-1}(\gamma\zeta\zeta)=\{abb,abbc\}$. We consider two cases (1) $abb$ and (2) $abbc$.
  For case (1), we split $abb$ to $ab$ and $b$ because $\ell_2(ab)=\gamma\zeta$ and $\ell_2(b)=\zeta$. $\ell_1(ab)=\alpha$, $\ell_1(b)=\alpha$, then one possibility of the bygone-state estimate of $G$ done by agent $A_1$ inferred by agent $A_2$ is $\Mt_1(G,\alpha,\alpha\alpha)=\dt_{\obs}^1(A,\alpha) \cap \dt_{\obs}^{1,\inv}(E,\alpha) = B\cap F = \{2\}$.
  For case (2), we split $abbc$ to $ab$ and $bc$ because $\ell_2(ab)=\gamma\zeta$ and $\ell_2(bc)=\zeta$. $\ell_1(ab)=\alpha$, $\ell_1(bc)=\alpha\alpha$, then the other possibility of the bygone-state estimate of $G$ done by agent $A_1$ inferred by agent $A_2$ is $\Mt_1(G,\alpha,\alpha\alpha\alpha)=\dt_{\obs}^1(A,\alpha) \cap \dt_{\obs}^{1,\inv}(E,\alpha\alpha) = B\cap F = \{2\}$. Then, we conclude that $\Mt_{A_1\leftarrow A_2}(G,\gamma\zeta,\gamma\zeta\zeta)=\{\{2\}\}$. This result is consistent with the above result obtained by using item~\eqref{item4:High-OrderProperty} of \autoref{thm13:High-OrderProperty}.

  By definition as in \eqref{eqn46:High-OrderProperty}, one sees that $Y_2\Cap Y_2^{\inv}$ is a reachable state of the
  concurrent composition $\overline{\CC}_{\ell_{\ep}} \left( \Obs_{A_1\leftarrow A_2}(G), \Obs_{A_1\leftarrow A_2}^{\inv}(G) \right)$. Moreover, $Y_2\Cap Y_2^{\inv}$ only contains a singleton $\{2\}$,
  then by \autoref{thm13:High-OrderProperty}, $G$ is not order-$2$ bygone-state opaque with respect to agents $A_1$ and $A_2$. Next we directly verify this result by definition.
  Consider the run $0 \xrightarrow[]{a} 1 \xrightarrow[]{b} 2 \xrightarrow[]{b} 3$. We
  consider $\ell_1(ab)=\alpha$, $\ell_1(b)=\alpha$. As shown before, $\Mt_1(G,\alpha,\alpha\alpha)=\{2\}$ which is a singleton. Also, we have $\ell_2(ab)=\gamma\zeta$ and $\ell_2(b)=\zeta$. We have shown that $\Mt_{A_1\leftarrow A_2}(G,\gamma\zeta,\gamma\zeta\zeta)=\{\{2\}\}$, which
  violates \autoref{def:order2-ByGState-Ooacity}, that is, we again obtain that $G$ is not order-$2$ bygone-state opaque with respect to agents $A_1$ and $A_2$.

  \begin{figure}[!htbp]
  \centering
  \subcaptionbox{FSA $G $ considered in \cite{Zhang2024High-OrderPropertyDES}, where $E_1=\{b,c,d\}$,
  $\ell_1(b) = \ell_1(c) = \alpha$, $\ell_1(d) = \beta$,
  $E_2=\{a,b\}$, $\ell_2(a) = \gamma$, $\ell_2(b) = \zeta$.\label{fig1:High-OrderCSO}}{
  	  \begin{tikzpicture}[>=stealth',shorten >=1pt,auto,node distance=2.0 cm, scale = 1.0, transform shape,
	>=stealth,inner sep=2pt]

	\node[initial, initial where = left, state] (0) {$0$};
	\node[state] (1) [right of =0] {$1$};
	\node[state] (2) [right of =1] {$2$};
	\node[state] (3) [right of =2] {$3$};
	\node[state] (4) [above right of =3] {$4$};
	\node[state] (5) [below right of =3] {$5$};

	\path [->]
	(0) edge node [above, sloped] {$a$} (1)
	(1) edge node [above, sloped] {$b$} (2)
	(2) edge [bend left] node {$b$} (3)
	(3) edge [bend left] node {$c$} (2)
	(3) edge node [above, sloped] {$a$} (4)
	(3) edge node [above, sloped] {$a$} (5)
	(4) edge [loop right] node {$d$} (4)
	;

    \end{tikzpicture}
  }\vspace{0.5cm}

  \subcaptionbox{Part of the observer $\Obs_{A_1}$ of LFSA $G_{A_1}$, where $A=\{0,1\}$, $B=\{2\}$, $C=\{3,4,5\}$, $D=\{4\}$. 
  \label{fig2:High-OrderCSO}}{
	\begin{tikzpicture}[>=stealth',shorten >=1pt,auto,node distance=2.5 cm, scale = 1.0, transform shape,
	>=stealth,inner sep=2pt]

	\node[initial, initial where = left, rectangular state] (01) {$A$};
	\node[rectangular state] (2) [right of =01] {$B$};
	\node[rectangular state] (345) [right of =2] {$C$};
	\node[rectangular state] (4) [right of =345] {$D$};

	\path [->]
	(01) edge node {$\alpha$} (2)
	(2) edge [bend left] node {$\alpha$} (345)
	(345) edge [bend left] node {$\alpha$} (2)
	(345) edge node {$\beta$} (4)
	(4) edge [loop right] node {$\beta$} (4)
	;

    \end{tikzpicture}
  }

  \subcaptionbox{Part of the concurrent composition $\CC(G_{A_1},\Obs_{A_1})$,
	where $A,B,C,D$ are as above.\label{fig3:High-OrderCSO}}{
	\begin{tikzpicture}[>=stealth',shorten >=1pt,auto,node distance=2.5 cm, scale = 1.0, transform shape,
	>=stealth,inner sep=2pt]

	\node[initial, initial where = left, rectangular state] (0A) {$(0,A)$};
	\node[rectangular state] (1A) [below =1cm of 0A] {$(1,A)$};
	\node[rectangular state] (2B) [right of =1A] {$(2,B)$};
	\node[rectangular state] (3C) [right of =2B] {$(3,C)$};
	\node[rectangular state] (4C) [above right of =3C] {$(4,C)$};
	\node[rectangular state] (5C) [below right of =3C] {$(5,C)$};
	\node[rectangular state] (4D) [right of =4C] {$(4,D)$};

	\path [->]
	(0A) edge node {$(a,\ep)$} (1A)
	(1A) edge node {$(b,\alpha)$} (2B)
	(2B) edge [bend left] node {$(b,\alpha)$} (3C)
	(3C) edge [bend left] node {$(c,\alpha)$} (2B)
	(3C) edge node [sloped, above] {$(a,\ep)$} (4C)
	(3C) edge node [sloped, above] {$(a,\ep)$} (5C)
	(4C) edge node {$(d,\beta)$} (4D)
	(4D) edge [loop right] node {$(d,\beta)$} (4D)
	;

    \end{tikzpicture}
  }\vspace{0.5cm}

  \subcaptionbox{Part of the order-$2$ observer 
	$\Obs_{A_1\leftarrow A_2}(G )$,
	where $A,B,C,D$ are as above.
 	\label{fig4:High-OrderCSO}}{
	\begin{tikzpicture}[>=stealth',shorten >=1pt,auto,node distance=2.5 cm, scale = 1.0, transform shape,>=stealth,inner sep=2pt, align=center]

	  \node[initial, initial where = left, rectangular state] (0A) {\begin{tabular}{c}$Y_0$\\\hline$\{(0,A)\}$\end{tabular}};
	  \node[rectangular state] (1A) [below =1cm of 0A] {\begin{tabular}{c}$Y_1$\\\hline$\{(1,A)\}$\end{tabular}};
	  \node[rectangular state] (2B) [right of =1A] {\begin{tabular}{c}$Y_2$\\\hline$\{(2,B)\}$\end{tabular}};
	  \node[rectangular state] (2B3C) [right =1cm of 2B] {\begin{tabular}{c}$Y_3$\\\hline$\{(2,B),(3,C)\}$\end{tabular}};
	  \node[rectangular state] (4C5C4D) [right =1cm of 2B3C] {\begin{tabular}{c}$Y_4$\\\hline$\{(4,C),(5,C),(4,D)\}$\end{tabular}};

	\path [->]
	(0A) edge node {$\gamma$} (1A)
	(1A) edge node {$\zeta$} (2B)
	(2B) edge node {$\zeta$} (2B3C)
	(2B3C) edge node {$\gamma$} (4C5C4D)
	(2B3C) edge [loop above] node {$\zeta$} (2B3C)
	;

    \end{tikzpicture}
  }\vspace{0.5cm}

  \caption{illustrative automata in \autoref{exam11:High-OrderProperty}.}
  \end{figure}

  \begin{figure}[!htbp]
  \centering
  \subcaptionbox{FSA $G^{\inv}$, where $G$ is shown in \autoref{fig1:High-OrderCSO}.\label{fig1:High-OrderProperty}}{
  	  \begin{tikzpicture}[>=stealth',shorten >=1pt,auto,node distance=2.0 cm, scale = 1.0, transform shape,
	>=stealth,inner sep=2pt]

	\node[initial, initial where = above, state] (0) {$0$};
	\node[initial, initial where = above, state] (1) [right of =0] {$1$};
	\node[initial, initial where = above, state] (2) [right of =1] {$2$};
	\node[initial, initial where = above, state] (3) [right of =2] {$3$};
	\node[initial, initial where = above, state] (4) [above right of =3] {$4$};
	\node[initial, initial where = above, state] (5) [below right of =3] {$5$};

	\path [->]
	(1) edge node [above, sloped] {$a$} (0)
	(2) edge node [above, sloped] {$b$} (1)
	(3) edge [bend right] node [above] {$b$} (2)
	(2) edge [bend right] node [below] {$c$} (3)
	(4) edge node [above, sloped] {$a$} (3)
	(5) edge node [above, sloped] {$a$} (3)
	(4) edge [loop right] node {$d$} (4)
	;

    \end{tikzpicture}
  }\vspace{0.5cm}
  \subcaptionbox{Part of the inverse observer $\Obs_{A_1}^{\inv}$ of LFSA $G_{A_1}$, where 
	$G_{A_1}$ is shown in \autoref{fig1:High-OrderCSO}, $E=\{0,1,2,3,4,5\}$, $F=\{0,1,2,3\}$, $G=\{3,4\}$, $B=\{2\}$, $H=\{0,1,3\}$. 
  \label{fig2:High-OrderProperty}}{
	\begin{tikzpicture}[>=stealth',shorten >=1pt,auto,node distance=2.5 cm, scale = 1.0, transform shape,
	  >=stealth,inner sep=2pt, align=center]

	\node[initial, initial where = left, rectangular state] (E) {$E$};
	\node[rectangular state] (F) [above right of =E] {$F$};
	\node[rectangular state] (G) [below right of =E] {$G$};
	\node[rectangular state] (B) [above right of =G] {$B$};
	\node[rectangular state] (H) [right of =B] {$H$};

	\path [->]
	(E) edge node {$\alpha$} (F)
	(E) edge node {$\beta$} (G)
	(F) edge [loop above] node {$\alpha$} (F)
	(G) edge node {$\alpha$} (B)
	(G) edge [loop below] node {$\beta$} (G)
	(B) edge [bend left] node {$\alpha$} (H)
	(H) edge [bend left] node {$\alpha$} (B)
	;

    \end{tikzpicture}
  }

  \caption{illustrative automata in \autoref{exam11:High-OrderProperty}.}
\end{figure}

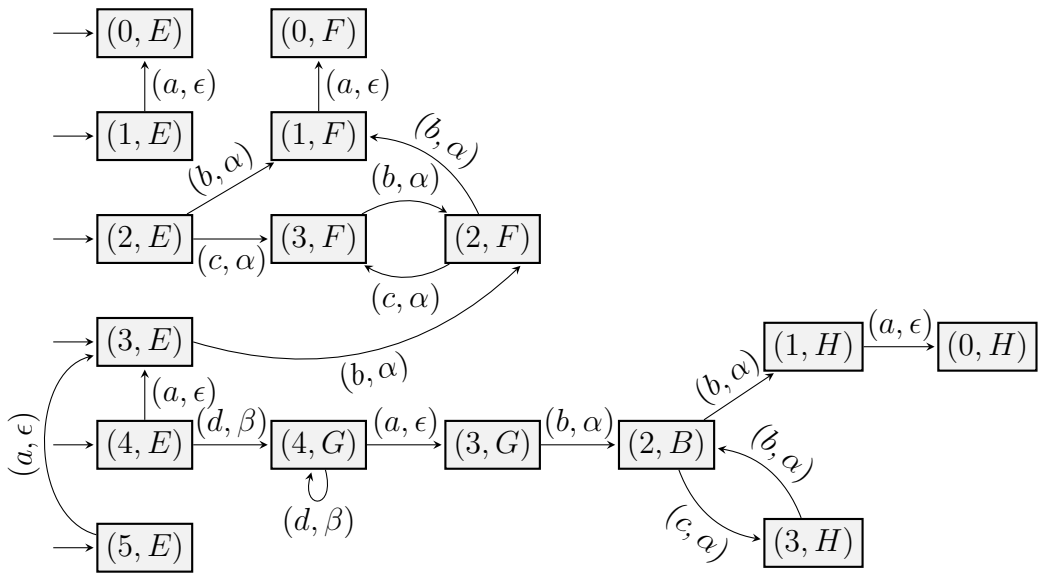
\begin{figure}[!htbp]
  \centering
	\begin{tikzpicture}[>=stealth',shorten >=1pt,auto,node distance=2.3cm, scale = 1.0, transform shape,
	>=stealth,inner sep=2pt]
	\tikzset{rectangular state/.style={draw, rectangle, thick, fill=gray!10, inner sep = 3pt}}

	\node[initial, initial where = left, rectangular state] (0E) {$(0,E)$};
	\node[initial, initial where = left, rectangular state] (1E) [below =0.7cm of 0E] {$(1,E)$};
	\node[initial, initial where = left, rectangular state] (2E) [below =0.7cm of 1E] {$(2,E)$};
	\node[initial, initial where = left, rectangular state] (3E) [below =0.7cm of 2E] {$(3,E)$};
	\node[initial, initial where = left, rectangular state] (4E) [below =0.7cm of 3E] {$(4,E)$};
	\node[initial, initial where = left, rectangular state] (5E) [below =0.7cm of 4E] {$(5,E)$};

	\node[rectangular state] (1F) [right of =1E] {$(1,F)$};
	\node[rectangular state] (0F) [above = 0.7cm of 1F] {$(0,F)$};
	\node[rectangular state] (3F) [below = 0.7cm of 1F] {$(3,F)$};
	\node[rectangular state] (2F) [right of =3F] {$(2,F)$};

	\node[rectangular state] (4G) [right of =4E] {$(4,G)$};
	\node[rectangular state] (3G) [right of =4G] {$(3,G)$};
	\node[rectangular state] (2B) [right of =3G] {$(2,B)$};
	\node[rectangular state] (1H) [above right =0.9cm of 2B] {$(1,H)$};
	\node[rectangular state] (3H) [below right =0.9cm of 2B] {$(3,H)$};
	\node[rectangular state] (0H) [right of = 1H] {$(0,H)$};

	\path [->]
	(1E) edge node [right] {$(a,\ep)$} (0E)
	(2E) edge node [sloped, above] {$(b,\alpha)$} (1F)
	(1F) edge node [right] {$(a,\ep)$} (0F)
	(2E) edge node [sloped, below] {$(c,\alpha)$} (3F)
	(2F) edge [bend right] node [sloped, above] {$(b,\alpha)$} (1F)
	(3F) edge [bend left] node [sloped, above] {$(b,\alpha)$} (2F)
	(2F) edge [bend left] node [sloped, below] {$(c,\alpha)$} (3F)
	(3E.0) edge [bend right] node [sloped, below] {$(b,\alpha)$} (2F.-40)
	(4E) edge node [right] {$(a,\ep)$} (3E)
	(4E) edge node {$(d,\beta)$} (4G)
	(4G) edge node {$(a,\ep)$} (3G)
	(4G) edge [loop below] node {$(d,\beta)$} (4G)
	(3G) edge node {$(b,\alpha)$} (2B)
	(2B) edge node [sloped, above] {$(b,\alpha)$} (1H)
	(2B) edge [bend right] node [sloped, below] {$(c,\alpha)$} (3H)
	(3H) edge [bend right] node [sloped, above] {$(b,\alpha)$} (2B)
	(1H) edge node {$(a,\ep)$} (0H)
	(5E) edge [bend left=75] node [sloped, above] {$(a,\ep)$} (3E)
	;

    \end{tikzpicture}
	\caption{Part of the concurrent composition $\CC\left(G_{A_1}^{\inv},\Obs_{A_1}^{\inv}\right)$,
	where $G_{A_1}^{\inv}$ is in \autoref{fig1:High-OrderProperty}, 
    $\Obs_{A_1}^{\inv}$ is in \autoref{fig2:High-OrderProperty}, $E=\{0,1,2,3,4,5\}$,
    $F=\{0,1,2,3\}$, $G=\{3,4\}$, $B=\{2\}$, $H=\{0,1,3\}$.}
    \label{fig3:High-OrderProperty}
  \end{figure}

  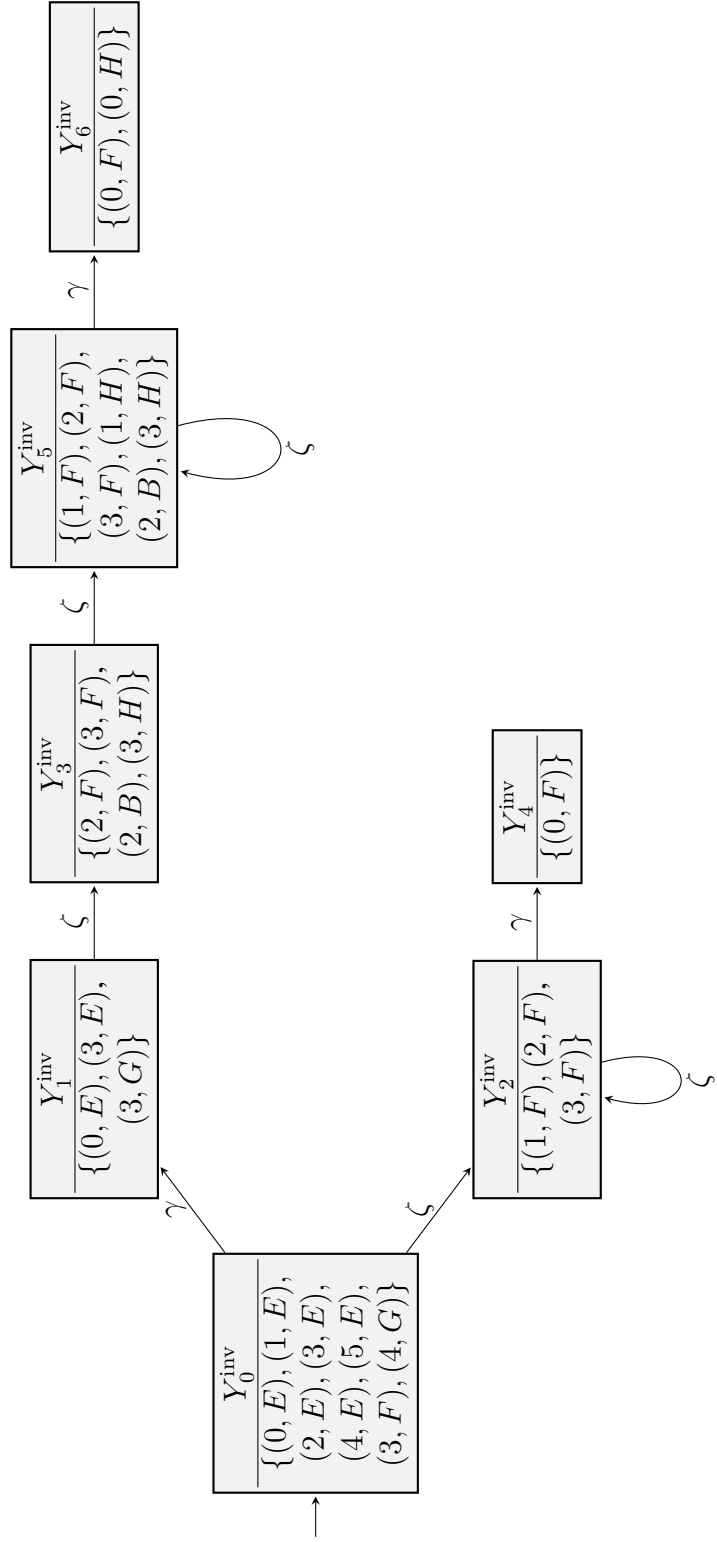
\begin{figure}[!htbp]
	\centering
	\begin{tikzpicture}[>=stealth',shorten >=1pt,auto,node distance=2.5 cm, scale = 1.0, transform shape,
	  >=stealth,inner sep=2pt, rotate=90, align=left]

	  \node[initial, initial where = left, rectangular state] (Y0inv) {\begin{tabular}{c}$Y_0^{\inv}$\\\hline$\{(0,E),(1,E),$\\$(2,E),(3,E),$\\$(4,E),(5,E),$\\$(3,F),(4,G)\}$\end{tabular}};
	  \node[rectangular state] (Y1inv) [above right =1cm of Y0inv] {\begin{tabular}{c}$Y_1^{\inv}$\\\hline$\{(0,E),(3,E)$,\\$(3,G)\}$\end{tabular}};
	  \node[rectangular state] (Y2inv) [below right =1cm of Y0inv] {\begin{tabular}{c}$Y_2^{\inv}$\\\hline$\{(1,F),(2,F)$,\\$(3,F)\}$\end{tabular}};
	  \node[rectangular state] (Y3inv) [right =1cm of Y1inv] {\begin{tabular}{c}$Y_3^{\inv}$\\\hline$\{(2,F),(3,F)$,\\$(2,B),(3,H)\}$\end{tabular}};
	  \node[rectangular state] (Y4inv) [right =1cm of Y2inv] {\begin{tabular}{c}$Y_4^{\inv}$\\\hline$\{(0,F)\}$\end{tabular}};
	  \node[rectangular state] (Y5inv) [right =1cm of Y3inv] {\begin{tabular}{c}$Y_5^{\inv}$\\\hline$\{(1,F),(2,F),$\\$(3,F),(1,H),$\\$(2,B),(3,H)\}$\end{tabular}};
	  \node[rectangular state] (Y6inv) [right =1cm of Y5inv] {\begin{tabular}{c}$Y_6^{\inv}$\\\hline$\{(0,F),(0,H)\}$\end{tabular}};

	\path [->]
	(Y0inv) edge node [above] {$\gamma$} (Y1inv)
	(Y0inv) edge node [above] {$\zeta$} (Y2inv)
	(Y1inv) edge node [above] {$\zeta$} (Y3inv)
	(Y2inv) edge node [above] {$\gamma$} (Y4inv)
	(Y3inv) edge node [above] {$\zeta$} (Y5inv)
	(Y5inv) edge node [above] {$\gamma$} (Y6inv)
	(Y2inv) edge [loop below] node {$\zeta$} (Y2inv)
	(Y5inv) edge [loop below] node {$\zeta$} (Y5inv)
	;

    \end{tikzpicture}
	\caption{Part of the order-$2$ inverse observer $\Obs_{A_1\leftarrow A_2}^{\inv}(G) =
	\Obs\left( \CC_{A_1\to A_2}^{G^{\inv},\Obs^{\inv}} \right)$, where 
	$\CC_{A_1\to A_2}^{G^{\inv},\Obs^{\inv}}$ is obtained
	from \autoref{fig3:High-OrderProperty}
	by changing each event $(*,**)$ to $*$ and then being endowed with the labeling 
	function $\ell_2$ as in \autoref{fig1:High-OrderCSO},
	$E=\{0,1,2,3,4,5\}$, $F=\{0,1,2,3\}$, $G=\{3,4\}$, $B=\{2\}$, $H=\{0,1,3\}$.}
 	\label{fig4:High-OrderProperty}
  \end{figure}

  \begin{figure}[!htbp]
	\centering
	\begin{tikzpicture}[>=stealth',shorten >=1pt,auto,node distance=3.0 cm, scale = 1.0, transform shape,
	  >=stealth,inner sep=2pt, rotate=0, align=left]
	  \tikzset{rectangular state/.style={draw, rectangle, thick, fill=gray!10, inner sep = 3pt}
}

	  \node[initial, initial where = left, rectangular state] (00) {$\left( Y_0,Y_0^{\inv} \right)$};
	  \node[rectangular state] (01) [right of = 00] {$\left( Y_0,Y_1^{\inv} \right)$};
	  \node[rectangular state] (10) [above of = 01] {$\left( Y_1,Y_0^{\inv} \right)$};
	  \node[rectangular state] (02) [below of = 01] {$\left( Y_0,Y_2^{\inv} \right)$};

	  \node[rectangular state] (12) [right of = 10] {$\left( Y_1,Y_2^{\inv} \right)$};
	  \node[rectangular state] (20) [above = 0.5cm of 12] {$\left( Y_2,Y_0^{\inv} \right)$};
	  \node[rectangular state] (11) [below = 0.5cm of 12] {$\left( Y_1,Y_1^{\inv} \right)$};
	  \node[rectangular state] (03) [right of = 01] {$\left( Y_0,Y_3^{\inv} \right)$};

	  \node[rectangular state] (04) [right of = 02] {$\left( Y_0,Y_4^{\inv} \right)$};

	  \node[rectangular state] (21) [right of = 20] {$\left( Y_2,Y_1^{\inv} \right)$};
	  \node[rectangular state] (30) [above = 0.5cm of 21] {$\left( Y_3,Y_0^{\inv} \right)$};
	  \node[rectangular state] (22) [below = 0.5cm of 21] {$\left( Y_2,Y_2^{\inv} \right)$};

	\path [->]
	(00) edge node [above] {$(\ep,\gamma)$} (01)
	(00) edge node [sloped, above] {$(\gamma,\ep)$} (10)
	(00) edge node [sloped, above] {$(\ep,\zeta)$} (02)
	(02) edge [loop below] node {$(\ep,\zeta)$} (02)

	(10) edge node [sloped, above] {$(\zeta,\ep)$} (20.180)
	(10) edge node [sloped, above] {$(\ep,\gamma)$} (11)
	(10) edge node [sloped, above] {$(\ep,\zeta)$} (12)
	(01) edge node [sloped, above] {$(\ep,\zeta)$} (03)
	(01) edge node [sloped, above] {$(\gamma,\ep)$} (11)
	(02) edge node [sloped, above] {$(\ep,\gamma)$} (04)
	(02) edge [bend left=13] node [sloped, above, pos=0.26] {$(\gamma,\ep)$} (12)
	
	(20) edge node [sloped, above] {$(\zeta,\ep)$} (30)
	(20) edge node [sloped, above] {$(\ep,\gamma)$} (21)
	(20) edge node [sloped, above] {$(\ep,\zeta)$} (22)
	(12) edge node [sloped, above] {$(\zeta,\ep)$} (22)
	;

    \end{tikzpicture}
	\caption{Part of the concurrent composition $\CC_{\ell_{\ep}} \left( \Obs_{A_1\leftarrow A_2}(G), \Obs_{A_1\leftarrow A_2}^{\inv}(G) \right)$, 
	  where the order-$2$ observer $\Obs_{A_1\leftarrow A_2}(G)$ is in \autoref{fig4:High-OrderCSO}, the order-$2$ inverse observer $\Obs_{A_1\leftarrow A_2}^{\inv}(G)$ is in 
	\autoref{fig4:High-OrderProperty}.}
 	\label{fig5:High-OrderProperty}
  \end{figure}
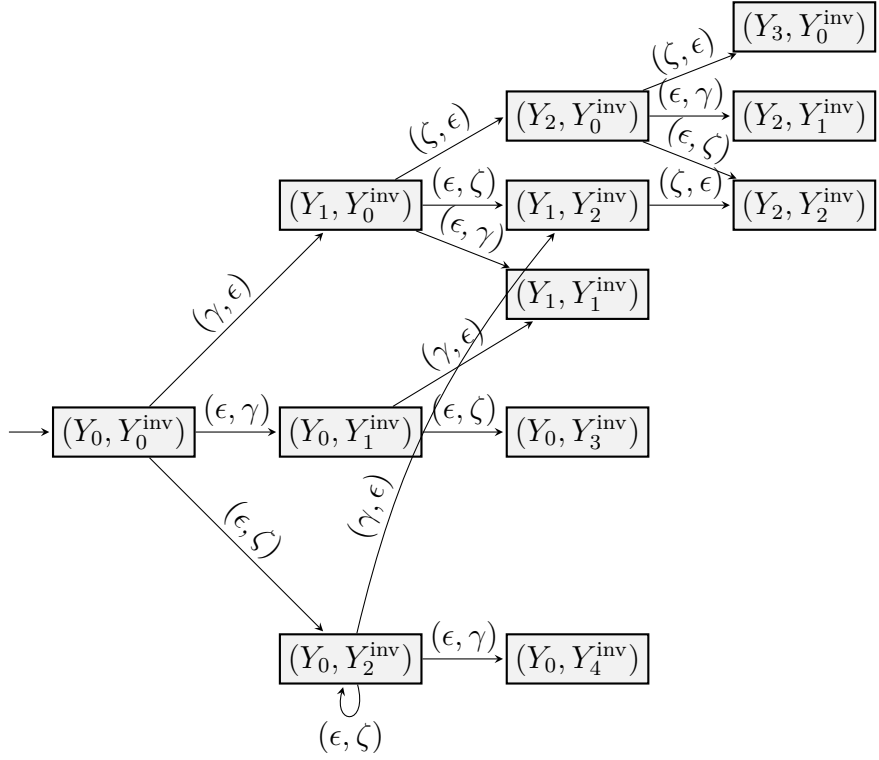

\end{example}

\section{Conclusion}

Given a finite-state automaton $G$ publicly known to agents
$A_1$ and $A_2$, 
assuming that each agent has its own observable event set of $G$ and $A_2$ knows
$A_1$'s observable events and labeling function, a notion of 
order-$2$ bygone-state opacity was
formulated to characterize what agent $A_2$ knows about 
$A_1$'s bygone-state estimate of $G$. Based on our concurrent composition and the
classical observer, we derived a tool to verify order-$2$ bygone-state opacity
in doubly exponential time. From the procedure to derive the tool, one will 
see that the steps are very natural. However, without following this procedure,
it will be very difficult to see whether one can derive another tool that
can verify order-$2$ bygone-state opacity. Such a formulation is closely related 
to an active area in mathematical logic and mathematical philosophy called
dynamic epistemic logic which characterizes knowledge exchange between 
different agents such as ``whether agent $A_1$ knows what/how/why agent $A_2$ knows''
\cite{DynamicEpistemicLogic2008,Baltag2025TopologyOfSurprise,Li2024Higher-orderEpistemicLogic}
but in the latter area, formulations are often 
done by logical fragments. Relations and differences between these two types of 
representations for dynamical information exchanges between agents are worthy of
further exploration.



\end{document}